\documentclass[10pt, draftcls, onecolumn]{IEEEtran}

\usepackage{amsmath,amssymb,amsthm,graphicx,enumerate,verbatim,xcolor}
\usepackage{bbm}
\usepackage{subcaption}
\usepackage{tikz}
\usetikzlibrary{arrows,positioning,fit,backgrounds}
\usetikzlibrary{calc}

\newcommand{\eps}{\varepsilon}

\newcommand{\Ebb}{\mathbb{E}}

\newcommand{\Pbb}{\mathbb{P}}
\newcommand{\Rbb}{\mathbb{R}}
\newcommand{\Acal}{\mathcal{A}}

\newcommand{\Ccal}{\mathcal{C}}
\newcommand{\Ecal}{\mathcal{E}}
\newcommand{\Lcal}{\mathcal{L}}
\newcommand{\Mcal}{\mathcal{M}}

\newcommand{\Scal}{\mathcal{S}}
\newcommand{\Tcal}{\mathcal{T}}

\newcommand{\Xcal}{\mathcal{X}}
\newcommand{\Ycal}{\mathcal{Y}}
\newcommand{\Zcal}{\mathcal{Z}}
\newcommand{\Xsf}{\mathsf{X}}
\newcommand{\xsf}{\mathsf{x}}

\newcommand{\ysf}{\mathsf{y}}

\newcommand{\zsf}{\mathsf{z}}

\DeclareMathOperator*{\argmin}{arg\,min}
\DeclareMathOperator*{\argmax}{arg\,max}

\newtheorem{thm}{Theorem}
\newtheorem{cor}{Corollary}
\newtheorem{lemma}{Lemma}
\newtheorem{defn}{Definition}
\newtheorem{prop}{Proposition}
\newtheorem{property}{Property}

\tikzstyle{arw}=[->,>=latex]
\tikzstyle{node}=[draw,rectangle,rounded corners, minimum width=1cm,minimum height =.75 cm]

\title{The Henchman Problem: Measuring Secrecy by the Minimum Distortion in a List}
\date{October, 2014}
\author{
\IEEEauthorblockN{Curt Schieler and Paul Cuff \thanks{C. Schieler (schieler@princeton.edu) and P. Cuff (cuff@princeton.edu) are with the Department of Electrical Engineering, Princeton University, Princeton, NJ, 08544. \newline This work was supported in part by the National Science Foundation under Grants CCF-1116013 and by the Air Force Office of Scientific Research under Grant FA9550-12-1-0196. Part of this work was presented in \cite{Schieler2014henchman}. }}
}
\begin{document}
\maketitle
\begin{abstract}
We introduce a new measure of information-theoretic secrecy based on rate-distortion theory and study it in the context of the Shannon cipher system. Whereas rate-distortion theory is traditionally concerned with a single reconstruction sequence, in this work we suppose that an eavesdropper produces a list of $2^{nR_{\sf L}}$ reconstruction sequences and measure secrecy by the minimum distortion over the entire list. We show that this setting is equivalent to one in which an eavesdropper must reconstruct a single sequence, but also receives side information about the source sequence and public message from a rate-limited henchman (a helper for an adversary). We characterize the optimal tradeoff of secret key rate, list rate, and eavesdropper distortion. The solution hinges on a problem of independent interest: lossy compression of a codeword drawn uniformly from a random codebook. We also characterize the solution to the lossy communication version of the problem in which distortion is allowed at the legitimate receiver. The analysis in both settings is greatly aided by a recent technique for proving source coding results with the use of a likelihood encoder.
\end{abstract}
\section{Introduction}
A ubiquitous model in information-theoretic secrecy is the Shannon cipher system \cite{Shannon1949} in which two nodes who share secret key want to communicate losslessly in the presence of an eavesdropper. As depicted in Figure~\ref{fig:scs}, Node A views an i.i.d. source sequence $X^n$ and uses the shared secret key $K$ that is independent of the source to produce an encrypted message $M$. Node B uses the message and the key to produce $\hat{X}^n$. An eavesdropper views the message and knows the scheme that Nodes A and B employ.

Also ubiquitous is the investigation of how to measure secrecy when there is not enough key to ensure perfect secrecy, i.e. when the key rate is less than the entropy of the information source. One potential solution, proposed by Yamamoto in \cite{Yamamoto1997}, is to measure secrecy by the distortion that an eavesdropper incurs in attempting to reconstruct the source sequence. In accordance with the usual constructs in rate-distortion theory, this means that Nodes A and B want to maximize the following expression over all possible codes:
\begin{equation}
\min_{z^n(m)} \Pbb[d(X^n,z^n(M)) \geq D].
\end{equation}
Although this seems like a reasonable objective at first glance, it was shown in \cite{Schieler2013} that simple codes employing negligible rates of secret key can force this probability to one, regardless of the distortion level $D$. The reason for this disconcerting result is that the accompanying secrecy guarantees can be fragile, as the following example elucidates. Let $X^n$ be i.i.d. $\text{Bern}(1/2)$ and suppose that there is just one bit of secret key, i.e. $K\in\{0,1\}$. Encrypt by transmitting $X^n$ itself if $K=0$ and $X^n$ with all its bits flipped if $K=1$. In this scenario, any optimal reconstruction $Z^n$ that the eavesdropper produces has expected hamming distortion equal to $1/2$, the highest expected distortion that the eavesdropper could possibly incur. Despite this, the eavesdropper actually knows quite a bit about $X^n$, namely that it is one of two sequences. Indeed, the guarantee of secrecy is rather fragile because if the eavesdropper learns just one bit of the source sequence, then the entire sequence is compromised. 

In view of the previous example, one way to strengthen a distortion-based measure of secrecy is to design schemes around the assumption that the eavesdropper has access to some side information. In \cite{Schieler2013}, this is accomplished by supposing that eavesdropper views the causal behavior of the system; in particular, the eavesdropper reconstructs $Z_i$ based on $X^{i-1}$ and the public message $M$.

\begin{figure}
\begin{center}
\begin{tikzpicture}
 [node distance=1cm,minimum width=1cm,minimum height =.75 cm]
  \node[rectangle,minimum width=5mm] (source) {$X^n$};
  \node[node] (alice) [right =7mm of source] {A};
  \node[node] (bob) [right =3cm of alice] {B};
  \node[coordinate] (dummy) at ($(alice.east)!0.5!(bob.west)$) {};
  \node[rectangle,minimum width=5mm] (xhat) [right =7mm of bob] {$\widehat{X}^n$};
  \node[rectangle,minimum width=7mm] (key) [above =7mm of dummy] {$K\in[2^{nR_0}]$};
  \node[node] (eve) [below =3mm of bob] {Eve};
  
  \draw [arw] (source) to (alice);
  \draw [arw] (alice) to node[minimum height=6mm,inner sep=0pt,midway,above]{$M\in[2^{nR}]$} (bob);
  \draw [arw] (bob) to (xhat);
  \draw [arw] (key) to [out=180,in=90] (alice);
  \draw [arw] (key) to [out=0,in=90] (bob);
  \draw [arw,rounded corners] (dummy) |- (eve);
 \end{tikzpicture}
 \caption{\small The Shannon cipher system with secret key rate $R_0$ and communication rate $R$. In this paper, we measure secrecy by the minimum distortion in a list of reconstruction sequences $\{Z^n(1),\ldots,Z^n(2^{nR_{\sf L}})\}$ that the eavesdropper produces.}
 \label{fig:scs}
 \end{center}
 \end{figure}
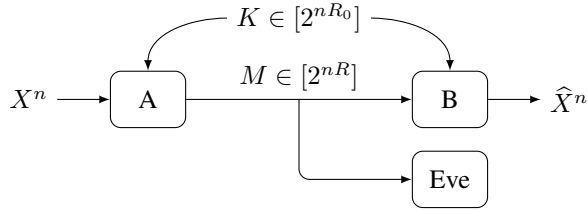

In this paper, we study another distortion-based approach to measuring secrecy in the Shannon cipher system. Instead of requiring a single reconstruction sequence $Z^n$, we suppose that the eavesdropper produces a list of $2^{nR_{\sf L}}$ reconstructions $\{Z^n(1),\ldots,Z^n(2^{nR_{\sf L}})\}$ and consider the minimum distortion over the entire list. This is somewhat reminiscent of equivocation (i.e.,  the conditional entropy $H(X^n|M)$), which also purports to measure the uncertainty of the eavesdropper. However, an important difference in the measure we study is that the structure of the uncertainty is built directly into the definition. The eavesdropper's equivocation merely provides a lower bound on the size of the smallest list that contains the exact source sequence $X^n$. On the other hand, the optimal tradeoff between secret key rate, distortion, and list rate will give us a function $R_{\sf L}(R_0,D)$ that precisely quantifies the size of the smallest list that an eavesdropper is able to produce that reliably contains a sequence of distortion $D$. 

Quantifying secrecy in terms of lists and distortion has been done previously in \cite{Merhav1999} and \cite{Haroutunian2010}, where the eavesdropper is modeled as a ``guessing wiretapper" who produces a sequence of reconstructions. After each estimate, the eavesdropper receives feedback about whether or not the reconstruction was within a certain distortion level.\footnote{In \cite{Merhav1999}, the feedback concerns exact reconstructions, whereas \cite{Haroutunian2010} allows a distortion parameter.} As soon as the distortion level is reached, the eavesdropper stops guessing; the moments of the number of guesses needed indicate the secrecy of the system. Our approach differs from these works in that there is no sequential guessing (no testing mechanism) and the list size is fixed.
\subsection*{Organization}
This paper considers the list-reconstruction measure of secrecy and establishes the information-theoretic characterization of the optimal tradeoffs among the secret key rate, list rate, and distortion at the eavesdropper. We divide the paper into two parts. First, we introduce and solve the problem when lossless communication is required between the legitimate parties (Sections~\ref{sec:setup}--\ref{sec:losslessachievability}). We then introduce the lossy communication setting and solve the corresponding problem (Sections~\ref{sec:mainlossy}--\ref{sec:lossyachievability}), reusing components from the preceding sections where possible. Although the lossy communication setting is a generalization of the lossless setting, there are several complications and subtleties that emerge that warrant the separation. For example, the converse proof is much more involved in the lossy setting.

In Section~\ref{sec:setup}, we formally define the list-based measure of secrecy and the lossless communication setting in which it will be first be analyzed.  We also give an equivalent reformulation of the setting in terms of a malicious helper for the eavesdropper; the resulting ``henchman problem" becomes the default formulation for the remainder of the paper. Section~\ref{sec:mainlossless} contains Theorem~\ref{mainresult}, the characterization of the optimal tradeoffs in the lossless communication setting. The proof of Theorem~\ref{mainresult} is presented in Section~\ref{sec:losslessconverse} (converse) and Section~\ref{sec:losslessachievability} (achievability). In Section~\ref{sec:mainlossy}, we introduce the lossy communication version of the problem and characterize the optimal tradeoffs in Theorem~\ref{mainresultlossy}. The converse and achievability proofs of Theorem~\ref{mainresultlossy} are given in Sections~\ref{sec:lossyconverse} and \ref{sec:lossyachievability}, respectively.

In addition to being a treatment of a new measure of secrecy for the Shannon cipher system, this paper is an endorsement of the efficacy of a likelihood encoder for proving source coding results. As detailed in \cite{Song2014}, a likelihood encoder is a particular stochastic encoder which, when combined with a random codebook, manages to avoid many of the tedious and technical components of achievability proofs in lossy compression problems. The primary conduit for the analysis of a likelihood encoder is the ``soft covering lemma", which is expounded upon in \cite{Cuff2013}. In our case, the technique allows us to extract an idealized subproblem from the crucial part of the achievability proof and consider it independently of the original problem. The subproblem concerns the lossy compression of a codeword drawn uniformly from a random codebook.

\section{Preliminaries}
\label{sec:setup}
\subsection{Notation}
All alphabets (e.g., $\Xcal$, $\Ycal$, and $\Zcal$) are finite. The set $\{1,\ldots,m\}$ is sometimes denoted by $[m]$. Given a per-letter distortion measure $d(x,z)$, we abuse notation slightly by defining
\begin{equation}
d(x^n,z^n) \triangleq \frac1n \sum_{i=1}^n d(x_i,z_i).
\end{equation}
We also assume that for every $x\in\Xcal$, there exists $z\in\Zcal$ such that $d(x,z)=0$. 

We denote the empirical distribution (or type) of a sequence $x^n$ by $T_{x^n}$:
\begin{equation}
T_{x^n}(x) = \frac1n \sum_{i=1}^n \mathbf{1}\{x = x_i\}.
\end{equation}

\subsection{Total variation distance}
Throughout the paper, we make frequent use of the total variation distance between two probability measures $P$ and $Q$ with common alphabet, defined by
\begin{equation}
\lVert P - Q \rVert_{\sf TV} \triangleq \sup_{A\in\mathcal{F}} |P(A) - Q(A)|.  
\end{equation}
The following properties of total variation distance are quite useful.
\begin{property}
\label{tvproperties}
Total variation distance satisfies:
\begin{enumerate}[(a)]
\item If the support of $P$ and $Q$ is a countable set $\Xcal$, then
\begin{equation}
\lVert P - Q \rVert_{\sf TV} = \frac12 \sum_{x\in\Xcal} |P(\{x\})-Q(\{x\})|.
\end{equation}
\item Let $\eps>0$ and let $f(x)$ be a function with bounded range of width $b>0$. Then
\begin{equation}
\label{tvcontinuous}
\lVert P-Q \rVert_{\sf TV} < \eps \:\Longrightarrow\: \big| \Ebb_Pf(X) - \Ebb_Qf(X) \big | < \eps b,
\end{equation}
where $\Ebb_{P}$ indicates that the expectation is taken with respect to the distribution $P$.
\item Let $P_{X}P_{Y|X}$ and $Q_XP_{Y|X}$ be two joint distributions with common channel $P_{Y|X}$. Then
\begin{equation}
\lVert P_XP_{Y|X} - Q_X P_{Y|X} \rVert_{\sf TV} = \lVert P_X - Q_X \rVert_{\sf TV}.
\end{equation}
\item Let $P_X$ and $Q_X$ be marginal distributions of $P_{XY}$ and $Q_{XY}$. Then
\begin{equation}
\lVert P_X - Q_X \rVert_{\sf TV} \leq \lVert P_{XY} - Q_{XY} \rVert_{\sf TV}.
\end{equation}
\end{enumerate}
\end{property}

\subsection{Problem setup}
As shown in Figure~\ref{fig:scs}, Node A observes a source sequence $X^n$ that is i.i.d. according to a distribution $P_X$. Nodes A and B share common randomness $K\in[2^{nR_0}]$ that is uniformly distributed and independent of $X^n$. Node A sends a message $M$ to Node B over a noiseless channel at rate $R$. 
\begin{defn}
An $(n,R,R_0)$ code consists of:
\begin{IEEEeqnarray}{sl}
Encoder: & f:\Xcal^n\times[2^{nR_0}]\rightarrow [2^{nR}]\\
Decoder: & g:[2^{nR}]\times[2^{nR_0}]\rightarrow \Xcal^n
\end{IEEEeqnarray}
The encoder and decoder can be stochastic (in which case they are denoted by $P_{M|X^n,K}$ and $P_{\widehat{X}^n|M,K}$).
\end{defn}

The encrypted communication (the message $M$) is overheard perfectly by an eavesdropper who produces a list $\Lcal(M)\subset \Zcal^n$ and incurs the minimum distortion over the entire list:
\begin{equation}
\min _{z^n\in \Lcal(M)} d(X^n,z^n).
\end{equation}
Using the secret key and the noiseless channel, Nodes A and B want to communicate losslessly while ensuring that the eavesdropper's optimal strategy suffers distortion above a given level with high probability. The generalization to lossy communication begins in Section~\ref{sec:mainlossy}.

\begin{defn}
\label{listdefn}
The tuple $(R,R_0,R_{\sf L},D)$ is achievable if there exists a sequence of $(n,R,R_0)$ codes such that the error probability $\Pbb[X^n \neq \widehat{X}^n]$ vanishes and, $\forall \eps > 0$,
\begin{equation}
\min_{\substack{\Lcal(m):|\Lcal|\leq2^{nR_{\sf L}}}} \Pbb\Big[\min_{z^n \in \Lcal(M)} d(X^n,z^n) \geq D-\eps\Big] \xrightarrow{n\to\infty} 1.
\end{equation}
\end{defn}
 Thus, we allow the eavesdropper to use any list-valued function $\Lcal:\Mcal \rightarrow \{\Zcal^n\}_1^{2^{nR_{\sf L}}}$, provided the cardinality of the range satisfies {$|\Lcal|~\leq~2^{nR_{\sf L}}$}. Furthermore, we assume that the eavesdropper knows the $(n,R,R_0)$ code and the distribution $P_X$.
 
 \subsection{The henchman problem}
 So far, the problem has been formulated in terms of an eavesdropper who produces a list of $2^{nR_{\sf L}}$ reconstructions. It turns out that we can relate this formulation to one in which an eavesdropper reconstructs a single sequence; this is accomplished by supplying the eavesdropper with a rate-limited helper (a henchman). As depicted in Figure~\ref{fig:henchman}, the eavesdropper receives $nR_{\sf L}$ bits of side information from a henchman who has access to the source sequence $X^n$ and the public message $M$. Since the eavesdropper and henchman cooperate, this means that the eavesdropper effectively receives the best possible $nR_{\sf L}$ bits of side information about the pair $(X^n,M)$ to assist in producing a single reconstruction sequence $Z^n$. 

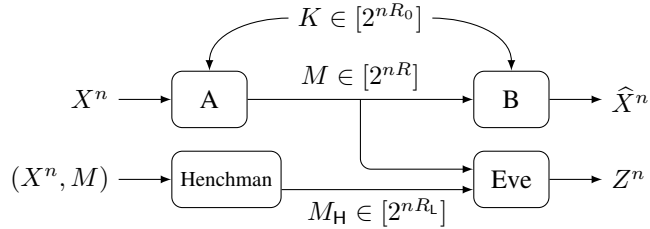
\begin{figure}
\begin{center}
\begin{tikzpicture}
 [node distance=1cm,minimum width=1cm,minimum height =.75 cm]
  \node[rectangle,minimum width=5mm] (source) {$X^n$};
  \node[node] (alice) [right =7mm of source] {A};
  \node[node] (bob) [right =3cm of alice] {B};
  \node[coordinate] (dummy) at ($(alice.east)!0.5!(bob.west)$) {};
  \node[rectangle,minimum width=5mm] (xhat) [right =7mm of bob] {$\widehat{X}^n$};
  \node[rectangle,minimum width=7mm] (key) [above =7mm of dummy] {$K\in[2^{nR_0}]$};
  \node[node] (eve) [below =3mm of bob] {Eve};

  \node[node,anchor=west] (hman) at (eve -| alice.west) {\footnotesize Henchman};
  \node[rectangle,minimum width=5mm] (hsource) [left=7mm of hman] {$(X^n,M)$};
  \node[rectangle,minimum width=5mm] (z) [right=7mm of eve] {$Z^n$};

  \draw [arw] (source) to (alice);
  \draw [arw] (alice) to node[minimum height=6mm,inner sep=0pt,midway,above]{$M\in[2^{nR}]$} (bob);
  \draw [arw] (bob) to (xhat);
  \draw [arw] (key) to [out=180,in=90] (alice);
  \draw [arw] (key) to [out=0,in=90] (bob);
  \draw [arw,rounded corners] (dummy) |- (eve.165);
  
  \draw [arw] (hman.east |- eve.195) to node[minimum height=6mm,inner sep=0pt,midway,below]{$M_{\sf H}\in[2^{nR_{\sf L}}]$} (eve.195);
  \draw [arw] (hsource) to (hman);
  \draw [arw] (eve) to (z);
 \end{tikzpicture}
 \caption{\small The henchman problem. A rate-limited henchman has access to the source sequence and the public message. The eavesdropper produces a single reconstruction sequence $Z^n$ based on the public message and the side information from the henchman.}
 \label{fig:henchman}
 \end{center}
 \end{figure}

\begin{defn}
\label{henchmandefn}
The tuple $(R,R_0,R_{\sf L},D)$ is achievable in the henchman problem if there exists a sequence of $(n,R,R_0)$ codes such that the error probability $\Pbb[X^n \neq \widehat{X}^n]$ vanishes and, $\forall \eps > 0$,
\begin{equation}
\label{mainobj}
\min_{\substack{m_{\sf H}(x^n,m),z^n(m,m_{\sf H}):\\|\Mcal_{\sf H}|\leq2^{nR_{\sf L}}}} \Pbb\Big[d(X^n,z^n(M,M_{\sf H})) \geq D-\eps\Big] \xrightarrow{n\to\infty} 1.
\end{equation}
\end{defn}
Thus, we allow the eavesdropper and henchman to jointly design a code consisting of an encoder $m_{\sf H}(x^n,m)$ and a decoder $z^n(m,m_{\sf H})$, subject to the constraint $|\Mcal_{\sf H}| \leq 2^{nR_{\sf L}}$. It can be shown that allowing a stochastic encoder or decoder does not decrease the eavesdropper's distortion. As in Definition~\ref{listdefn}, we assume that the adversarial entities are aware of the scheme that Nodes A and B employ, although this is not explicitly indicated in \eqref{mainobj}.

We now demonstrate the equivalence of the list reconstruction problem and the henchman problem.

\begin{prop}
The tuple $(R,R_0,R_{\sf L},D)$ is achievable in the list reconstruction problem if and only if it is achievable in the henchman problem. In other words, Definitions~\ref{listdefn} and~\ref{henchmandefn} are equivalent.
\end{prop}
\begin{proof}
It is enough to show that the eavesdropper's scheme in the list reconstruction problem can be transformed to a scheme in the henchman problem that achieves the same (or less) distortion, and vice versa. 

Let $\Lcal(m)$ be the function that the eavesdropper uses to produce a list of reconstruction sequences. If the public message is $M$, the list $\Lcal(M)$ can act as a codebook in the henchman problem. Knowing $(X^n,M)$, the henchman can transmit the index of the sequence in $\Lcal(M)$ with the lowest distortion. Upon receiving the index and $M$, the eavesdropper reconstructs the corresponding sequence.

Conversely, suppose that the henchman and eavesdropper have devised an encoder $m_{\sf H}(x^n,m)$ and a decoder $z^n(m,m_{\sf H})$. Upon observing the public message, the eavesdropper has a list of codewords (one for each $m_{\sf H}$) that can be used for the list reconstruction problem. More precisely, the eavesdropper forms the list
\begin{equation}
\Lcal(M) = \{z^n(M,m_{\sf H})\}_{m_{\sf H} \in [2^{nR_{\sf L}}]}.
\end{equation}
In both cases, it is straightforward to verify that the transformation maintains (or decreases) the distortion. To carry out the verification formally, it is enough to show that for any $(n,R,R_0)$ code, 
\begin{equation}
\min_{\substack{\Lcal(m):|\Lcal|\leq2^{nR_{\sf L}}}} \Pbb\Big[\min_{z^n \in \Lcal(M)} d(X^n,z^n) \geq D\Big]  =  \min_{\substack{m_{\sf H}(x^n,m),z^n(m,m_{\sf H}):\\|\Mcal_{\sf H}|\leq2^{nR_{\sf L}}}} \Pbb\Big[d(X^n,z^n(M,M_{\sf H})) \geq D\Big].
\end{equation}
To show $(\geq)$, fix a list reconstruction function $\Lcal(m)$ and define a henchman encoder and eavesdropper decoder by
\begin{IEEEeqnarray}{rCl}
m_{\sf H}(x^n,m) &=& \argmin_{j\in[2^{nR_{\sf L}}]} d(x^n,\Lcal(m,j))\\
z^n(m,m_{\sf H}) &=& \Lcal(m,m_{\sf H}),
\end{IEEEeqnarray}
where $\Lcal(m,j)$ denotes the $j$th element of the list $\Lcal(m)$. Then we have
\begin{IEEEeqnarray}{rCl}
 \Pbb\Big[\min_{z^n \in \Lcal(M)} d(X^n,z^n) \geq D\Big] &=&  \Pbb\Big[d(X^n,z^n(M,M_{\sf H})) \geq D\Big]\\
 &\geq& \min_{\substack{m_{\sf H}(x^n,m),z^n(m,m_{\sf H}):\\|\Mcal_{\sf H}|\leq2^{nR_{\sf L}}}} \Pbb\Big[d(X^n,z^n(M,M_{\sf H})) \geq D\Big].
\end{IEEEeqnarray}
To show $(\leq)$, fix a henchman encoder $m_{\sf H}(x^n,m)$ and eavesdropper decoder $z^n(m,m_{\sf H})$ and define a list reconstruction function by
\begin{equation}
\Lcal(m) = \{z^n(m,m_{\sf H})\}_{m_{\sf H}\in[2^{nR_{\sf L}}]}.
\end{equation}
Then we have
\begin{IEEEeqnarray}{rCl}
\Pbb\Big[d(X^n,z^n(M,M_{\sf H})) \geq D\Big] &\geq& \Pbb\Big[\min_{z^n \in \Lcal(M)} d(X^n,z^n) \geq D\Big]\\
&\geq& \min_{\substack{\Lcal(m):|\Lcal|\leq2^{nR_{\sf L}}}} \Pbb\Big[\min_{z^n \in \Lcal(M)} d(X^n,z^n) \geq D\Big].
\end{IEEEeqnarray}
\end{proof}

\section{Main Result (lossless communication)}
\label{sec:mainlossless}
When lossless communication is required between the legitimate parties, we have the following characterization of the tradeoff among the communication rate, secret key rate, list rate (or henchman rate), and eavesdropper's distortion.
\begin{thm}
\label{mainresult}
Given a source distribution $P_X$ and a distortion function $d(x,z)$, the closure of achievable tuples $(R,R_0,R_{\sf L},D)$ is the set of tuples satisfying
 \begin{equation}
 \begin{IEEEeqnarraybox}[][c]{rCl}
 R &\geq& H(X)\\
 D &\leq& D(R_{\sf L})\cdot \mathbf{1}\{R_0>R_{\sf L}\},
 \end{IEEEeqnarraybox}
\end{equation}
where $D(\cdot)$ is the point-to-point distortion-rate function:
\begin{equation}
D(R) \triangleq \min_{P_{Z|X}:R\geq I(X;Z)} \Ebb[d(X,Z)].
\end{equation}
\end{thm}


Perhaps the most striking part of Theorem~\ref{mainresult} is that the region is discontinuous. Fixing a rate of secret key $R_0$, observe that when the list rate $R_{\sf L}$ is strictly less than $R_0$, the $(R_{\sf L},D)$ tradeoff follows the point-to-point rate-distortion function. However, as soon as $R_{\sf L}$ equals or exceeds the secret key rate, the eavesdropper's distortion drops to zero (the minimum distortion possible) because all possible decryptions can be enumerated in a list of size $2^{nR_0}$. Figure~\ref{fig:region} illustrates Theorem~\ref{mainresult} for a $\text{Bern}(1/2)$ source and hamming distortion; the communication rate is assumed to satisfy $R\geq H(X)$ and has no effect on the $(R_0,R_{\sf L},D)$ tradeoff.

Note that setting $R_{\sf L}=0$ in the region of Theorem~\ref{mainresult} corresponds to requiring a single reconstruction (without a henchman), which was Yamamoto's original formulation of the problem in \cite{Yamamoto1997}. In this case, we see that any positive rate of secret key results in distortion $D(0)$, the maximum expected distortion that can occur.

In the context of the list reconstruction formulation, Theorem~\ref{mainresult} implies that when Nodes A and B act optimally and $R_{\sf L} < R_0$, the eavesdropper's best strategy is to simply ignore the public message and list the codewords from a good point-to-point rate-distortion codebook. In particular, the public message is useless to the eavesdropper in this regime. However, when $R_{\sf L} \geq R_0$, the eavesdropper uses a different strategy and produces all possible decryptions of the public message. When we consider the lossy communication setting, we will see that a similar strategy switch occurs.

We now prove the achievability and converse portions of Theorem~\ref{mainresult}. For the entirety of the proof, we use the henchman formulation instead of the list reconstruction one. The main idea in the proof of achievability concerns the problem of compressing codewords from a random codebook beyond the rate-distortion limit; the proof also relies on a likelihood encoder \cite{Song2014} and the soft covering lemma {\cite[Lemma IV.1]{Cuff2013}}. The converse is straightforward, as we now show.

\begin{figure}
\begin{center}
\includegraphics{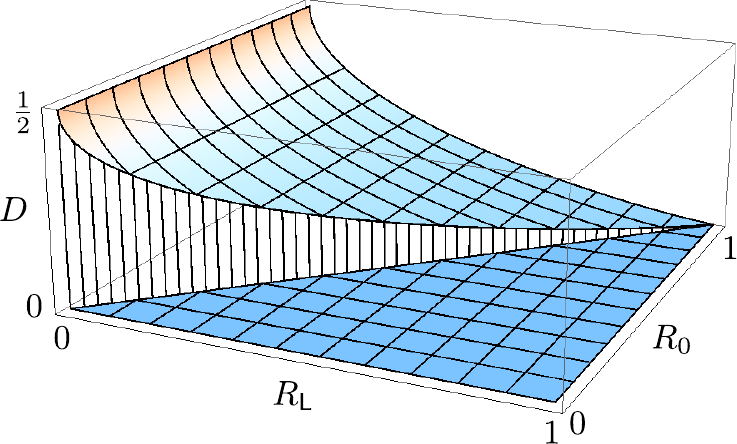}
\caption{\small The region in Theorem~\ref{mainresult} for source distribution {$P_X~\sim~\text{Bern}(1/2)$} and distortion measure $d(x,z)=\mathbf{1}\{x\neq z\}$.}
\label{fig:region}
\end{center}
\end{figure}

\section{Converse (Lossless communication)}
\label{sec:losslessconverse}
The constraint $R\geq H(X)$ is a consequence of the lossless source coding theorem. The constraint on $D$ splits into two cases depending on the relation between $R_0$ and $R_{\sf L}$. If $R_{\sf L} \geq R_0$, then any scheme that Nodes A and B use to achieve lossless compression can be exploited by the eavesdropper and the henchman. Since they can both enumerate the $2^{nR_0}$ possible decryptions of $M$, the henchman can simply send the index of the correct decryption, which results in zero distortion. On the other hand, if $R_{\sf L} < R_0$ then the eavesdropper and the henchman can ignore $M$ altogether and simply use a point-to-point rate-distortion code to describe $X^n$ within distortion $D(R_{\sf L})$. Therefore, regardless of the code that Alice and Bob use for lossless communication, the eavesdropper and the henchman can achieve distortion less than or equal to $D(R_{\sf L})\cdot \mathbf{1}\{R_0>R_{\sf L}\}$.

\section{Achievability (Lossless communication)}
\label{sec:losslessachievability}
Viewing the problem from the perspective of the adversarial entities, we see that the henchman observes the pair $(X^n,M)$ and encodes a message $M_{\sf H}$, and the eavesdropper observes $(M,M_{\sf H})$ and decodes $Z^n$; their goal is to minimize the distortion $d(X^n,Z^n)$. This describes the usual rate-distortion setting with additional information $M$ available at encoder and decoder. In other words, for a given $M=m$, the henchman observes a sequence $X^n$ drawn from a source distribution $P_{X^n|M=m}$ and describes the sequence to the eavesdropper using a rate-limited channel; the conditional distribution $P_{X^n|M=m}$ is the effective source distribution because both the henchman and the eavesdropper know the public message.

Observing that $P_{X^n|M=m}$ is induced entirely by the actions of Node A, let us assume for the moment that Node A uses the following random binning scheme to encode the source sequence. First, randomly divide the set of typical $x^n$ sequences into bins of size $2^{nR_0}$. This binning is known to everyone, including the adversaries. To encode $X^n$, Node A transmits the message {$M=(M_p,M_s)$}, where $M_p$ is the bin containing $X^n$, and $M_s$ is the index within that bin, one-time padded with $K$. Note that the one-time pad renders $M_s$ statistically independent of $X^n$ and $M_s$. Thus, for this choice of encoder, the induced distribution $P_{X^n|M}$ corresponds to choosing a sequence roughly uniformly at random from bin $M_p$ (because of the asymptotic equipartition property). Furthermore, the asymptotic equipartition property and the randomness of the binning suggest that the $2^{nR_0}$ sequences in bin $M_p$ were approximately chosen i.i.d. according to $\prod_{i=1}^n P_X(x_i)$. Therefore, very roughly speaking, the random binning scheme results in a distribution $P_{X^n|M=m}$ that corresponds to selecting a sequence uniformly from a random codebook whose codewords are generated independently and identically according to $\prod_{i=1}^n P_{X}(x_i)$. If this is true, then the joint goal of the henchman and the eavesdropper becomes the following: lossy compression (at rate $R_{\sf L}$) of a codeword drawn uniformly from a random codebook of size $2^{nR_0}$. We now delve into this subproblem, the conclusion of which is the following: if $R_{\sf L}<R_0$, then with high probability it is impossible to achieve distortion less than $D(R_{\sf L})$.

\subsection{Interlude: lossy compression of a codeword drawn uniformly from a random codebook}



Consider a codebook $c_{\xsf} = \{x^n(1),\ldots,x^n(2^{nR_{\sf C}})\}$ consisting of $2^{nR_{\sf C}}$ sequences. Select a codeword uniformly at random from $c_{\xsf}$ and denote it by $x^n(J)$, where $J \sim \text{Unif}[2^{nR_{\sf C}}]$. An encoder describes $x^n(J)$ using a noiseless link of rate $R$, and a decoder estimates it with a reconstruction sequence $Z^n$. Both the encoder and decoder know the codebook $c_{\xsf}$. Notice that the relationship between this setting and the standard rate-distortion framework is that the input space is contracted from $\Xcal^n$ to a codebook $c_{\xsf}$, and the source sequence is chosen uniformly at random from $c_{\xsf}$ instead of i.i.d. according to a distribution $P_X$.


\begin{defn}
For a fixed codebook $c_{\xsf} \subseteq \Xcal^n$, define an $(n,c_{\xsf},R)$ code as an encoder $f:\Xcal^n \times c_{\xsf} \rightarrow [2^{nR}]$ and a decoder $g:[2^{nR}]\times c_{\xsf} \rightarrow \Zcal^n$.  
\end{defn}
For a given $D$ and fixed codebook $c_{\xsf}$, the encoder and decoder want to maximize the probability that the distortion between $x^n(J)$ and $Z^n$ is less than $D$:
\begin{equation}
\label{cbdist}
\max_{(n,c_{\xsf},R)\text{ codes}} \Pbb[d(x^n(J),Z^n)\leq D].
\end{equation}
Instead of considering this objective for arbitrary codebooks, we generate a random codebook $\Ccal_{\xsf}$ in which the codewords are drawn independently, each according to $\prod_{i=1}^n P_X(x_i)$. The setup is depicted in Figure~\ref{fig:codebook}.

In some regimes, the expression in~\eqref{cbdist} approaches one as blocklength increases. For example, if $R\geq R_{\sf C}$ then the encoder can simply send the index of $X^n(J)$ within the codebook $\Ccal_{\xsf}$, thus ensuring zero distortion. Another example is when $R\geq R(D)$, in which case distortion $D$ is achievable even without knowledge of $\Ccal_{\xsf}$, because $X^n(J)$ is i.i.d. according to $P_X$. 

The regime we are interested in is when $R < R_{\sf C}$ and~{$R~<~R(D)$}. In this case, we find that with high probability (over the random codebook) it is impossible to achieve distortion $D$, i.e., the expression in \eqref{cbdist} vanishes.

\begin{thm}
\label{cbthm}
Fix $R$, $R_{\sf C}$ and $D$. Let $\Ccal_{\xsf}$ be a random codebook of $2^{nR_{\sf C}}$ codewords, each drawn independently according to $\prod_{i=1}^n P_X(x_i)$. Let $\tau_n$ be any sequence that converges to zero sub-exponentially fast (i.e., $\tau_n=2^{-o(n)}$). If
\begin{equation}
R<\min\{R(D),R_{\sf C}\},
\end{equation}
then
\begin{equation}
\label{cbobj}
\lim_{n\to\infty}\Pbb_{\Ccal_{\xsf}}\Big[ \max_{(n,\Ccal_{\xsf},R)\text{ codes}}  \Pbb[d(X^n(J),Z^n)\leq D] > \tau_n \Big] = 0.
\end{equation}
\end{thm}

\begin{proof}

We first provide a brief, informal sketch of the proof idea. For an optimal $(n,\Ccal_{\xsf},R)$ code, there are on average $2^{n(R_{\sf C} - R)}$ codewords in $\Ccal_{\xsf}$ that map to each of the $2^{nR}$ reconstruction sequences in $\Zcal^n$. However, for a given reconstruction sequence $z^n$, there are only (on average) $2^{n(R_{\sf C}-R(D))}$ sequences in $\Ccal_{\xsf}$ within distortion $D$ of $z^n$, because the probability of an i.i.d sequence $X^n$ being within distortion $D$ of $z^n$ is roughly $2^{-nR(D)}$. Since $2^{n(R_{\sf C}-R(D))}$ is much smaller than $2^{n(R_{\sf C} - R)}$, the probability that $z^n$ yields distortion less than D is vanishingly small. In fact, this probability decays doubly exponentially, which means that the entire suite of $2^{nR}$ reconstruction sequences simultaneously yields distortion greater than $D$ with high probability. In other words, the optimal code gives rise to distortion greater than $D$ with high probability, which is what we want to show.

The first step is to restrict $X^n(J)$ to the $\delta$-typical set $\Tcal_{\delta}^n(X)$ by writing
\begin{equation}
\label{typical} \Pbb[d(X^n(J),Z^n)\leq D] \leq \Pbb[d(X^n(J),Z^n)\leq D, \Acal] + \Pbb[\Acal^c],
\end{equation}
where $\Acal$ denotes the event $\{X^n(J)\in \Tcal_{\delta}^n \}$. The $\delta$-typical set is defined according to the notion of strong typicality:
\begin{equation}
\Tcal_{\delta}^n(X) \triangleq \{x^n\in\Xcal^n: \lVert T_{x^n} - P_X \rVert_{\sf TV} < \delta\},
\end{equation}
where $T_{x^n}$ denotes the empirical distribution (i.e., the type) of $x^n$.
We will choose an appropriate $\delta$ later. Note that the second term in~\eqref{typical} vanishes in the limit for any $\delta>0$ since $X^n(J)$ is i.i.d. according to $P_X$.

\begin{figure}
\begin{center}
\begin{tikzpicture}
 [node distance=1cm,minimum width=1cm,minimum height =.75 cm]
  \node[rectangle,minimum width=5mm] (source) {$X^n(J)$};
  \node[node] (alice) [right =7mm of source] {Enc.};
  \node[node] (bob) [right =2cm of alice] {Dec.};
  \node[coordinate] (dummy) at ($(alice.east)!0.5!(bob.west)$) {};
  \node[rectangle,minimum width=5mm] (xhat) [right =7mm of bob] {$Z^n$};
  \node[rectangle,minimum width=7mm] (key) [above =7mm of dummy] {$\Ccal_{\xsf}$};
  
  \draw [arw] (source) to (alice);
  \draw [arw] (alice) to node[minimum height=6mm,inner sep=0pt,midway,above]{$R$} (bob);
  \draw [arw] (bob) to (xhat);
  \draw [arw] (key) to [out=180,in=90] (alice);
  \draw [arw] (key) to [out=0,in=90] (bob);
 \end{tikzpicture}
 \caption{\small Lossy compression of a codeword drawn uniformly from a random codebook $\Ccal_{\xsf}=\{X^n(1),\ldots,X^n(2^{nR_{\sf C}})\}$. Both the encoder and decoder know the codebook $\Ccal_{\xsf}$, and the encoder must describe a randomly chosen codeword $X^n(J)$, where $J\sim \text{Unif}[2^{nR_{\sf C}}]$.}
 \label{fig:codebook}
 \end{center}
 \end{figure}
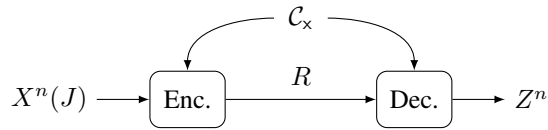

Although we defined a $(n,\Ccal_{\xsf},R)$ code as an encoder-decoder pair $(f,g)$, we will benefit from viewing a code as the combination of a codebook of $z^n$ sequences and an encoder that is optimal for that codebook. In other words, treat an $(n,\Ccal_{\xsf},R)$ code as a codebook $c_{\zsf}\subseteq \Zcal^n$ of size $2^{nR}$, together with an encoder that maps $x^n \in \Ccal_{\xsf}$ to the $z^n\in c_{\zsf}$ with the lowest distortion $d(x^n,z^n)$. This allows us to write
\begin{IEEEeqnarray}{l}
\nonumber \max_{(n,\Ccal_{\xsf},R)\text{ codes}}  \Pbb[d(X^n(J),Z^n)\leq D,\Acal] \\
\label{codeequiv}\quad\quad = \max_{c_{\zsf}(\Ccal_{\xsf})}  \Pbb\Big[\min_{z^n \in {c_{\zsf}(\Ccal_{\Xsf})} } d(X^n(J),z^n)\leq D,\Acal\Big],
\end{IEEEeqnarray}
where the notation ${c_{\zsf}(\Ccal_{\xsf})}$ emphasizes that $c_{\zsf}$ is a function of the random codebook $\Ccal_{\xsf}$; for simplicity, we suppress the $n$ and $R$ parameters of $c_{\zsf}(\Ccal_{\xsf})$.

Now we apply a union bound to the right-hand side of \eqref{codeequiv} and write
\begin{IEEEeqnarray}{rCl}
\Pbb\Big[\min_{z^n \in {c_{\zsf}(\Ccal_{\xsf})} } d(X^n(J),z^n)\leq D,\Acal\Big]
&\stackrel{(a)}{\leq}& \sum_{z^n\in c_{\zsf}(\Ccal_{\xsf})} \Pbb\Big[ d(X^n(J),z^n)\leq D,\Acal\Big]\\
&\leq& 2^{nR} \max_{z^n\in c_{\zsf}(\Ccal_{\xsf})} \Pbb\Big[ d(X^n(J),z^n)\leq D,\Acal\Big]\\
&\leq& 2^{nR} \max_{z^n\in \Zcal^n} \Pbb\Big[ d(X^n(J),z^n)\leq D,\Acal\Big]\\
\label{rvdenote} &\stackrel{(b)}{=}& 2^{-n(R_{\sf C} - R)} \max_{z^n\in \Zcal^n} \sum_{j=1}^{2^{nR_{\sf C}}} \mathbf{1}\{ d(X^n(j),z^n)\leq D, X^n(j) \in \Tcal_{\delta}^n \}, 
\end{IEEEeqnarray}
where step (a) is a union bound, and step (b) uses the fact that $X^n(J)$ is chosen uniformly from $\Ccal_{\xsf}$. Notice that for a fixed $z^n$, the terms in the sum in \eqref{rvdenote} are i.i.d. random variables (due to the nature of the random codebook construction), which we henceforth denote by $\xi_{j,z^n}$:
\begin{equation}
\label{xij}
\xi_{j,z^n} \triangleq \mathbf{1}\{ d(X^n(j),z^n)\leq D, X^n(j) \in \Tcal_{\delta}^n \},\quad j=1,\ldots,2^{nR_{\sf C}}.
\end{equation}


Using the equality in \eqref{codeequiv} and the bound in \eqref{rvdenote}, we have
\begin{IEEEeqnarray}{rCl}
\IEEEeqnarraymulticol{3}{l}{\nonumber
\Pbb\Big[ \max_{(n,\Ccal_{\xsf},R)\text{ codes}}  \Pbb[d(X^n(J),Z^n)\leq D,\Acal] > \tau_n \Big]
}\\
\quad&\leq& \Pbb\Big[\max_{z^n\in \Zcal^n} \sum_{j=1}^{2^{nR_{\sf C}}} \xi_{j,z^n} > \tau_n 2^{n(R_{\sf C}-R)}\Big] \\
\label{finalprob}&\stackrel{(a)}{\leq}& |\Zcal|^n \max_{z^n\in\Zcal^n} \Pbb\Big[\sum_{j=1}^{2^{nR_{\sf C}}} \xi_{j,z^n} > \tau_n 2^{n(R_{\sf C}-R)}\Big],
\end{IEEEeqnarray}
where (a) is a union bound. If we can show that the probability in \eqref{finalprob} decays doubly exponentially fast with $n$, then the proof will be complete. To that end, we first use a standard application of the method of types \cite{Csiszar1998} to establish a bound on the expected value of $\xi_{j,z^n}$ in the following lemma. The proof is relegated to the appendix.


\begin{lemma}
\label{typebound}
If $X^n$ is i.i.d. according to $P_X$, then for any $z^n$,
\begin{equation}
\Pbb[d(X^n,z^n)\leq D, X^n \in \Tcal_{\delta}^n] \leq 2^{-n(R(D)-o(1))},
\end{equation}
where $R(D)$ is the point-to-point rate-distortion function for $P_X$, and $o(1)$ is a term that vanishes as $\delta\rightarrow 0$ and $n\rightarrow \infty$.
\end{lemma}

From Lemma~\ref{typebound}, we see that the expected value of $\sum_{j=1}^{2^{nR_{\sf C}}} \xi_{j,z^n}$ is bounded above by approximately $2^{n(R_{\sf C}-R(D))}$. Moreover, since a condition of the theorem being proved is that $R<R(D)$, it follows that $ \tau_n 2^{n(R_{\sf C}-R)}$ is exponentially larger than $2^{n(R_{\sf C}-R(D))}$. Therefore, \eqref{finalprob} is concerned with the probability that a sum of $2^{nR_{\sf C}}$ i.i.d. Bernoulli random variables is exponentially far away from its mean. Such a probability decays at a doubly exponential rate, as the following Chernoff bound will imply.
\begin{lemma}
\label{chernoff}
If $X^m$ is a sequence of i.i.d. $\text{Bern}(p)$ random variables, then
\begin{equation}
\Pbb\Big[\sum_{i=1}^m X_i > k\Big] \leq \left(  \frac{e\!\cdot\! m \!\cdot\! p}{k} \right )^k.
\end{equation}
\end{lemma}

\begin{proof}
The proof follows some of the usual steps for establishing Chernoff bounds.
\begin{IEEEeqnarray}{rCl}
\label{chernoffcorstep}\Pbb\Big[\sum_{i=1}^m X_i > k\Big] &\leq& \min_{\lambda > 0} e^{- \lambda k} \prod_{i=1}^m \Ebb[e^{\lambda X_i}] \\
&=& \min_{\lambda > 0} e^{- \lambda k} (p\cdot e^{\lambda} + 1 - p)^m \\
&\leq& \min_{\lambda > 0} e^{- \lambda k} (p\cdot e^{\lambda} + 1)^m \\
&\leq& \min_{\lambda > 0} e^{- \lambda k} e^{mpe^{\lambda}}
\end{IEEEeqnarray}
Substituting the minimizer $\lambda^* = \ln(\frac{k}{mp})$ gives the desired bound.
\end{proof}


%

Using the bound on $\Ebb[\xi_{j,z^n}]$ from Lemma~\ref{typebound}, we can apply Lemma~\ref{chernoff} to the probability in \eqref{finalprob} by identifying
\begin{IEEEeqnarray}{rCl}
m &=& 2^{n R_{\sf C}}\\
p &\leq& 2^{-n(R(D)-o(1))}\\ 
k &=& \tau_n 2^{n (R_{\sf C}-R)}.
\end{IEEEeqnarray}
 This gives
\begin{equation}
\label{doubleexp}
\Pbb\Big[\sum_{j=1}^{2^{nR_{\sf C}}} \xi_{j,z^n} > \tau_n 2^{n(R_{\sf C}-R)}\Big] \leq 2^{-n\alpha2^{n\beta}},
\end{equation}
where
\begin{IEEEeqnarray}{rCl}
\alpha &=& R(D) - R -o(1)\\
\beta &=& R_{\sf C}-R-o(1).
\end{IEEEeqnarray}
For small enough $\delta$ and large enough $n$, both $\alpha$ and $\beta$ are positive and bounded away from zero, and \eqref{doubleexp} vanishes doubly exponentially fast. Consequently, the expression in \eqref{finalprob} vanishes, completing the proof of Theorem~\ref{cbthm}.
\end{proof}
One can readily establish the following corollary to Theorem~\ref{cbthm}. 
\begin{cor}
\label{cbcor}
If $R < R_{\sf C}$ and $R < R(D)$, then
\begin{equation}
\lim_{n\to\infty}\Ebb_{\Ccal_{\xsf}}\Big[ \min_{(n,\Ccal_{\xsf},R)\text{ codes}}  \Pbb[d(X^n(J),Y^n)\geq D]\Big] = 1.
\end{equation}
\end{cor}


The interlude is now complete, and we can return to the achievability proof of Theorem~\ref{mainresult}.

\subsection{Likelihood encoder}
Earlier, we asserted that a scheme similar to random binning might give rise to an induced distribution $P_{X^n|M=m}$ that could be approximated by drawing a codeword uniformly from a random codebook. Then we could apply Corollary~\ref{cbcor} to our problem by identifying $(R_{\sf C},R)$ with $(R_0,R_{\sf L})$. Although it is possible that an encoder using random binning might yield this distribution, we turn instead to a likelihood encoder with a random codebook because it brings considerable clarity to the induced distributions involved.

Consider a codebook $c=\{x^n(m,k)\}$ consisting of $2^{n(R+R_0)}$ sequences from $\Xcal^n$. The likelihood encoder of \cite{Song2014} for lossless reconstruction and for this codebook is a stochastic encoder defined by
\begin{equation}
P_{M|X^nK}(m|x^n,k) \propto \prod_{i=1}^n \mathbf{1}\{x_i = x_i(m,k)\},
\end{equation}
where $\propto$ indicates that appropriate normalization is required.\footnote{In the rare case that no codeword is equal to the source sequence, an arbitrary index can be chosen.} The merit of using a likelihood encoder with a random codebook is that the resulting system-induced joint distribution of $(X^n,M,K)$, namely $P_{X^nMK}=P_{X^n}P_{K}P_{M|X^nK}$, can be shown to be close to an idealized distribution $Q_{X^nMK}$ defined by
\begin{equation}
Q_{X^nMK}(x^n,m,k) \triangleq 2^{-n(R+R_0)}\prod_{i=1}^n \mathbf{1}\{x_i = x_i(m,k)\}.
\end{equation}
More precisely, one can use the soft covering lemma~{\cite[Lemma IV.1]{Cuff2013}} to prove the following.
\begin{lemma}
\label{distrclose}
Let $\Ccal=\{X^n(m,k)\},{(m,k)\in[2^{nR}]\times[2^{nR_0}]}$ be a random codebook with each codeword drawn independently according to $\prod_{i=1}^n P_X$. If $R>H(X)$, then
\begin{equation}
\label{distrtv}
\lim_{n\to\infty} \Ebb_{\Ccal} \big\lVert P_{X^nMK} - Q_{X^nMK} \big\rVert_{\sf TV} = 0,
\end{equation}
where the expectation is with respect to the random codebook and $\lVert \cdot \rVert_{\sf TV}$ is total variation distance.
\end{lemma}

\begin{proof}
From the definition of $P_{X^nMK}$ and $Q_{X^nMK}$ we have $P_{M|X^nK}=Q_{M|X^nK}$. Using this fact, we have
\begin{IEEEeqnarray}{rCl}
\Ebb_{\Ccal} \big\lVert P_{X^nMK} - Q_{X^nMK} \big\rVert_{\sf TV} &\stackrel{(a)}{=}& \Ebb_{\Ccal} \big\lVert P_{X^nK} - Q_{X^nK} \big\rVert_{\sf TV} \\
&=& \Ebb_{\Ccal} \big\lVert P_{X^n}P_{K} - Q_{X^n|K}P_{K} \big\rVert_{\sf TV}\\
&=& 2^{-nR_0} \sum_{k=1}^{2^{nR_0}} \Ebb_{\Ccal} \big\lVert P_{X^n} - Q_{X^n|K=k} \big\rVert_{\sf TV},
\end{IEEEeqnarray}
where (a) uses Property~\ref{tvproperties}c. Since $R > H(X)$, the soft covering lemma implies that the summands vanish\footnote{Furthermore, they vanish uniformly for all $k \in [2^{nR_0}]$.}:
\begin{equation}
\lim_{n\to\infty} \Ebb_{\Ccal} \big\lVert P_{X^n} - Q_{X^n|K=k} \big\rVert_{\sf TV} = 0.
\end{equation}
Without getting into the details of the soft covering lemma, it is worthwhile to briefly summarize the main idea. The lemma, which is expounded upon in \cite{Cuff2013}, i The soft covering lemma applies to the current proof because $Q_{X^n|K=k}$ is the output distribution induced by a memoryless channel acting on a random codebook of size $2^{nR}$, and $P_{X^n}$ is an i.i.d. distribution. Since we are considering lossless communication in this section, the relevant channel is the noiseless identity channel, and the relevant rate condition is $R>H(X)$.
\end{proof}

Lemma~\ref{distrclose} and the definition of total variation distance allow us to analyze the probability in~\eqref{mainobj} as if $Q_{X^nMK}$ were the true system-induced joint distribution instead of $P_{X^nMK}$. This is important because $Q_{X^n|M=m}$ is, as desired, uniform over a random codebook of size $2^{nR_0}$:
\begin{equation}
\label{idealprop}
Q_{X^n|M=m} = \text{Unif}\{X^n(m,1),\ldots,X^n(m,2^{nR_0})\}.
\end{equation}
To see the role of $Q_{X^n|M=m}$, first denote (for the sake of brevity) the event
\begin{equation}
\Ecal = \{ d(X^n, z^n(M,M_{\sf H})) \geq D(R_{\sf L}) - \eps \}.
\end{equation}
Our objective is to show that when {$R_{\sf L} < R_0$}, Nodes A and B can force the eavesdropper to incur distortion $D(R_{\sf L})$, i.e., there exists a sequence of codes that ensures \eqref{mainobj}.

Taking the expectation of \eqref{mainobj} with respect to a random codebook, we have
\begin{IEEEeqnarray}{rCl}
\IEEEeqnarraymulticol{3}{l}{\nonumber
\Ebb_{\Ccal} \Big[\min_{\substack{m_{\sf H}(x^n,m),z^n(m,m_{\sf H}):\\|\Mcal_{\sf H}|\leq2^{nR_{\sf L}}}} \Pbb_{P_{X^nM}}
[\Ecal]\Big]}\\
\qquad &\stackrel{(a)}{=}& \Ebb_{\Ccal} \Big[\min_{\substack{m_{\sf H}(x^n,m),z^n(m,m_{\sf H}):\\|\Mcal_{\sf H}|\leq2^{nR_{\sf L}}}} \Pbb_{Q_{X^nM}} [\Ecal] \Big] + o(1) \\
&=& \Ebb_{\Ccal} \Big[\min_{\substack{m_{\sf H}(x^n,m),z^n(m,m_{\sf H}):\\|\Mcal_{\sf H}|\leq2^{nR_{\sf L}}}} \Ebb \big[\Pbb_{Q_{X^nM}} [\Ecal | M]\big] \Big] + o(1)\\
&\stackrel{(b)}{=}& \label{relatecor} \Ebb_M \,\Ebb_{\Ccal} \Big[\min_{\substack{m_{\sf H}(x^n),z^n(m_{\sf H}):\\|\Mcal_{\sf H}|\leq2^{nR_{\sf L}}}} \Pbb_{Q_{X^n|M}} [\Ecal | M] \Big] + o(1).
\end{IEEEeqnarray}
In step (a), we use Lemma~\ref{distrclose} to change the underlying distribution from $P_{X^nM}$ to $Q_{X^nM}$ with vanishing penalty. Step (b) uses the fact that $M$ and $\Ccal$ are independent under $Q_{X^nM}$. These steps bring us to the problem considered in the recent interlude: we must show that the henchman and the eavesdropper cannot design a code that achieves distortion $D(R_{\sf L})$ for the ``source" $Q_{X^n|M=m}$. 

Suppose that we are in the regime {$R_{\sf L}<R_0$}. The expression in \eqref{relatecor} is exactly what is addressed by Corollary~\ref{cbcor}, because the conditional distribution $Q_{X^n|M=m}$ corresponds to selecting a codeword uniformly from a random codebook of size $2^{nR_0}$ (as noted in \eqref{idealprop}), and $R_{\sf L}$ is the rate of the message sent from the ``encoder" (henchman) to the ``decoder" (eavesdropper). Also, note that we are invoking the corollary with $D=D(R_{\sf L}) - \eps$; thus, $R_{\sf L} < R(D)$ is satisfied.  Hence, Corollary~\ref{cbcor} gives
\begin{equation}
\lim_{n\to\infty}\Ebb_{\Ccal} \Big[\min_{\substack{m_{\sf H}(x^n),z^n(m_{\sf H}):\\|\Mcal_{\sf H}|\leq2^{nR_{\sf L}}}} \Pbb_{Q_{X^n|M=m}} [\Ecal | M=m] \Big] = 1.
\end{equation}
The likelihood encoder also provides the required lossless communication between Nodes A and B; it is straightforward to show that~\eqref{distrtv} implies vanishing probability of error if the decoder is defined by $g(m,k)=x^n(m,k)$. Indeed,
\begin{IEEEeqnarray}{rCl}
\Ebb_{\Ccal}\, \Pbb[X^n \neq \widehat{X}^n] &=& \Ebb_{\Ccal}\, \Pbb_{P_{X^nMK}}[X^n \neq X^n(M,K)] \\
&\stackrel{(a)}{=}&  \Ebb_{\Ccal}\, \Pbb_{Q_{X^nMK}}[X^n \neq X^n(M,K)] + o(1)\\
&=& 0 + o(1),
\end{IEEEeqnarray}
where step (a) follows from Lemma~\ref{distrclose} and the definition of total variation distance. 

We can conclude that there exists a codebook such that the associated likelihood encoder ensures~\eqref{mainobj} and lossless communication, because both hold when averaged over random codebooks. This completes the achievability portion of the proof of Theorem~\ref{mainresult}.

\section{Lossy communication}
\label{sec:mainlossy}
We now generalize the problem to allow distortion at the legitimate receiver. As depicted in Figure~\ref{fig:scslossy}, the receiver produces a reconstruction sequence $Y^n$, whose distortion is measured by $d(X^n,Y^n)$. Since there are two distortion measures (one for the receiver and one for the eavesdropper), we will distinguish them by using subscripts ${\sf B}$ and ${\sf E}$.

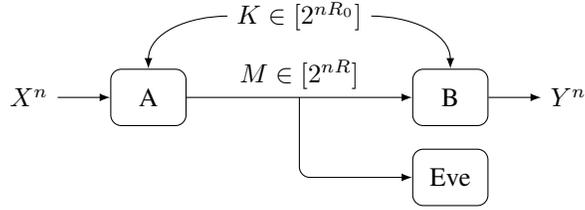
\begin{figure}
\begin{center}
\begin{tikzpicture}
 [node distance=1cm,minimum width=1cm,minimum height =.75 cm]
  \node[rectangle,minimum width=5mm] (source) {$X^n$};
  \node[node] (alice) [right =7mm of source] {A};
  \node[node] (bob) [right =3cm of alice] {B};
  \node[coordinate] (dummy) at ($(alice.east)!0.5!(bob.west)$) {};
  \node[rectangle,minimum width=5mm] (xhat) [right =7mm of bob] {$Y^n$};
  \node[rectangle,minimum width=7mm] (key) [above =7mm of dummy] {$K\in[2^{nR_0}]$};
  \node[node] (eve) [below =3mm of bob] {Eve};
  
  \draw [arw] (source) to (alice);
  \draw [arw] (alice) to node[minimum height=6mm,inner sep=0pt,midway,above]{$M\in[2^{nR}]$} (bob);
  \draw [arw] (bob) to (xhat);
  \draw [arw] (key) to [out=180,in=90] (alice);
  \draw [arw] (key) to [out=0,in=90] (bob);
  \draw [arw,rounded corners] (dummy) |- (eve);
 \end{tikzpicture}
 \caption{\small Lossy communication. Secrecy is measured by the minimum distortion in a list of reconstruction sequences $\{Z^n(1),\ldots,Z^n(2^{nR_{\sf L}})\}$ that the eavesdropper produces. There are two distortion functions at play, $d_{\sf B}(x,y)$ and $d_{\sf E}(x,z)$.}
 \label{fig:scslossy}
 \end{center}
 \end{figure}

\begin{defn}
\label{listdefnlossy}
The tuple $(R,R_0,R_{\sf L},D_{\sf B},D_{\sf E})$ is achievable if there exists a sequence of $(n,R,R_0)$ codes such that $\forall \eps > 0$,
\begin{enumerate}[1)]
\item Lossy communication:
\begin{equation}
\Pbb\Big[d_{\sf B}(X^n,Y^n) \leq D_{\sf B}+\eps\Big] \xrightarrow{n\to\infty} 1.
\end{equation}
\item List secrecy:
\begin{equation}
\label{listsecrecylossy}
\min_{\substack{\Lcal(m):|\Lcal|\leq2^{nR_{\sf L}}}} \Pbb\Big[\min_{z^n \in \Lcal(M)} d_{\sf E}(X^n,z^n) \geq D_{\sf E}-\eps\Big] \xrightarrow{n\to\infty} 1.
\end{equation}
\end{enumerate}
Note that the decoder of a $(n,R,R_0)$ code for lossy communication is a (possibly stochastic) decoder $P_{Y^n | MK}$.
\end{defn}

As with the lossless version of problem, we can reformulate Definition~\ref{listdefnlossy} in terms of a rate-limited henchman. The henchman formulation for the lossy communication setting is exactly described by Figure~\ref{fig:henchman} (with $\widehat{X}^n$ replaced by $Y^n$).

The optimal tradeoff between the various rates and distortions is the following.
\begin{thm}
\label{mainresultlossy}
Given a source distribution $P_X$ and distortion functions $d_{\sf B}(x,y)$ and $d_{\sf E}(x,z)$, the closure of achievable tuples $(R,R_0,R_{\sf L},D_{\sf B},D_{\sf E})$ is the set of tuples satisfying
 \begin{equation}
 \begin{IEEEeqnarraybox}[][c]{rCl}
 R &\geq& I(X;Y)\\
 D_{\sf B} &\geq& \Ebb\, d_{\sf B}(X,Y)\\
 D_{\sf E} &\leq& \begin{cases}
D(R_{\sf L}) & \text{if } R_{\sf L} < R_0\\
\min\{D(R_{\sf L}), D(R_{\sf L} - R_0,P_{XY})\} & \text{if }R_{\sf L} \geq R_0
\end{cases}
 \end{IEEEeqnarraybox}
\end{equation}
for some $P_{XY}=P_{X}P_{Y|X}$, where $D(\cdot,P_{XY})$ is the point-to-point distortion-rate function with side information channel $P_{Y|X}$ to the encoder and decoder:
\begin{equation}
D(R,P_{XY}) \triangleq \min_{P_{Z|XY}:R \geq I(X;Z|Y)} \Ebb\,d_{\sf E}(X,Z).
\end{equation}
\end{thm}
 
When $R_{\sf L} < R_0$, the eavesdropper's distortion is at least $D(R_{\sf L})$, just as it was when we considered lossless communication. This should not be surprising in light of the previous section, since less information is being revealed to the eavesdropper (the communication rate between Nodes A and B is lower). As before, the henchman can simply use a point-to-point rate-distortion code to achieve $D(R_{\sf L})$. 

The more interesting regime is when $R_{\sf L} \geq R_0$, i.e., the list rate (equivalently, the henchman's rate) is greater or equal to the rate of secret key. In this case, Theorem~\ref{mainresultlossy} says that a communication scheme can be designed such that the eavesdropper's distortion cannot be less than
\begin{equation}
\min\{D(R_{\sf L}), D_Y(R_{\sf L} - R_0)\}.
\end{equation}
To see why these are the relevant distortions, consider the following. As we just mentioned, the henchman and the eavesdropper can always ignore the message $M$ and use a point-to-point code to achieve $D(R_{\sf L})$. Alternatively, when $R_{\sf L} \geq R_0$, the henchman can first use part of the rate $R_{\sf L}$ to communication the secret key to the eavesdropper. Then, roughly speaking, the henchman and eavesdropper effectively share side information $Y^n$ (since they both know $M$ and $K$ perfectly and can mimic the decoder), and can use the remaining rate $R_{\sf L}-R_0$ to achieve distortion $D(R_{\sf L} - R_0,P_{XY})$. Thus, one implication of Theorem~\ref{mainresultlossy} is that the henchman benefits from sending information about the secret key only if he describes it entirely; there is no benefit to communicating just part of the key to the eavesdropper. 

\section{Converse (lossy communication)}
\label{sec:lossyconverse}
We now present the converse proof for Theorem~\ref{mainresultlossy}. In the regime $R_0 > R_{\sf L}$, the converse is the same as when we required lossless communication. Nodes A and B (the legitimate parties) cannot force distortion greater than $D(R)$ with high probability because the henchman and the eavesdropper can always ignore the public message $M$ and simply use a good rate-distortion code to achieve distortion $D(R)$ with high probability. Note that this converse is ``strong" in the sense that the probability of eavesdropper distortion being greater than $D(R)$ is not just bounded away from unity, it is actually vanishing. To be explicit, observe that if  $R_{\sf L} < R_0$ and $D_E > D(R)$, then the expression in \eqref{listsecrecylossy} vanishes for all $\eps < D_E - D(R)$. This follows from the achievability portion of point-to-point rate-distortion theory.

When $R_{\sf L} \geq R_0$, the henchman's rate is high enough that he can communicate the secret key to the eavesdropper and still have leftover rate $R_{\sf L} - R_0$. Since the henchman and the eavesdropper both know $M$ and $K$, they can mimic the decoder of Node B and produce side information $Y^n$. Notice that we have made two assumptions: the henchman knows the secret key, and the receiver uses a deterministic decoder. However, if the henchman were not able to determine the secret key exactly, then multiple keys would correspond to the same source sequence, which means that the decoder would effectively be stochastic. Thus, we are making just one assumption: that the decoder is deterministic. This assumption is valid because a stochastic decoder cannot be used to increase the eavesdropper's distortion. Indeed, if we consider the list formulation of the problem, we see that eavesdropper's performance is completely determined by $X^n$ and $M$ alone; the output of Node B does not play a role.\footnote{Note that this would not be the case if we were considering distortion functions of the form $d_{\sf E}(x,y,z)$ instead of $d_{\sf E}(x,z)$.}

So far, we have that the henchman and the eavesdropper share side information $Y^n$ equal to the receiver's reconstruction. Ideally, we would like to claim that $(X^n,Y^n)$ are jointly i.i.d. according to some distribution $P_{XY}$ and use the achievability portion of rate-distortion theory with side information at the encoder and decoder. Unfortunately, we cannot even claim that with high probability $(X^n,Y^n)$ are jointly typical according to some $P_{XY}$ because that is only guaranteed when Nodes A and B are using a nearly optimal rate-distortion code (i.e., one that operates near the rate-distortion tradeoff boundary). Instead, we rely on a different property of $(X^n,Y^n)$ that will be given shortly in Lemma~\ref{codetypes}.

We will describe the henchman and eavesdropper's scheme in terms of the joint type of $(X^n,Y^n)$; to do this, we require the following straightforward extension of the type-covering lemma \cite[Lemma 9.1]{Csiszar2011} that accounts for side information (proof omitted). Regarding notation, $\Tcal_{X}^n$ denotes the set of sequences whose types coincide with a given distribution $P_X$, and $\Tcal(\Xcal^n)$ denotes the set of all joint types on sequences in $\Xcal^n$.

\begin{lemma}
\label{typecovering}
Let $\tau > 0$ and $r\geq0$. Fix a joint type $P_{XY} \in \Tcal(\Xcal^n\times\Ycal^n)$, and let $y^n \in \Tcal_Y^n$. For $n \geq n_0(\tau)$, there exists a codebook $\Ccal(y^n, P_{XY}) \subseteq \Zcal^n$ such that
\begin{enumerate}[1)]
\item
\begin{equation}
\frac1n \log |\Ccal(y^n, P_{XY})| \leq r.
\end{equation}
\item For all $x^n$ such that $(x^n,y^n)\in \Tcal_{XY}^n$,
\begin{equation}
\min_{z^n\in\Ccal(y^n,P_{XY})} d(x^n,z^n) \leq D(r, P_{XY}) + \tau
\end{equation}
\end{enumerate}
\end{lemma}

We also require the following lemma from \cite{Weissman2005}.
\begin{lemma}[{\cite[Theorem 7]{Weissman2005}}]
\label{codetypes}
Consider any sequence of rate-distortion codes with rate $\leq R$. Then
\begin{equation}
\limsup_{n\to\infty} I(T_{X^nY^n}) \leq R \quad \text{a.s.},
\end{equation}
where $T_{X^nY^n}$ denotes the type of $(X^n,Y^n)$ and $I(\cdot)$ is the mutual information.
\end{lemma}

Now we can begin the converse proof for the regime $R_{\sf L} \geq R_0$. Consider an achievable tuple $(R,R_0,R_{\sf L},D_{\sf B},D_{\sf E})$. By the same argument that was used in the regime $R_{\sf L} < R_0$, we must have $D_E \leq D(R_{\sf L})$ because the henchman and the eavesdropper can always ignore the public message and use a good rate-distortion code to describe $X^n$. 

Let $\eps \in (0,1/2)$. Define a set $\Acal_n$ of joint distributions on $\Xcal\times\Ycal$ by
\begin{equation}
\Acal_n \triangleq \left\{
 \begin{IEEEeqnarraybox}[][c]{ll}
 Q_{XY}:\:\: & I_Q(X;Y) \leq R + \eps\\
 & \Ebb_Q\, d_{\sf B}(X,Y) \leq D_{\sf B} + \eps\\
 & \lVert Q_X - P_X \rVert_{\sf TV} \leq \eps
 \end{IEEEeqnarraybox}
\right\}.
\end{equation}
We first show that
\begin{equation}
\label{typepropertieslimit}
\lim_{n\to\infty} \Pbb[T_{X^nY^n} \in A_n] = 1,
\end{equation}
where $T_{X^nY^n}$ denotes the type of $(X^n,Y^n)$. This can be proved by combining the following three facts:
\begin{enumerate}[1)]
\item From Lemma~\ref{codetypes}, we have
\begin{equation}
\lim_{n\to\infty} \Pbb [I(T_{X^n,Y^n}) \leq R + \eps] = 1.
\end{equation}
\item From the definition of achievability and the equality $d_{\sf B}(x^n,y^n) = \Ebb_{T_{x^ny^n}}\,d_{\sf B}(x,y)$, we have
\begin{equation}
\lim_{n\to\infty} \Pbb [\Ebb_{T_{X^nY^n}}\,d_{\sf B}(x,y) \leq D + \eps] = 1.
\end{equation}
\item From the weak law of large numbers, we have
\begin{equation}
\lim_{n\to\infty} \Pbb [\lVert T_{X^n} - P_X \rVert_{\sf TV} \leq \eps] = 1.
\end{equation}
\end{enumerate}

With \eqref{typepropertieslimit} in hand, choose $n$ large enough so that
\begin{enumerate}[1)]
\item The eavesdropper cannot reconstruct with low distortion (this is an assumption of achievability):
\begin{equation}
\label{contradict}
\max_{\substack{m_{\sf H}(x^n,m),z^n(m,m_{\sf H}):\\|\Mcal_{\sf H}|\leq2^{nR_{\sf L}}}} \Pbb\Big[d_{\sf E}(X^n,z^n(M,M_{\sf H})) < D_{\sf E}-\eps\Big] < \eps.
\end{equation}
\item 
\begin{equation}
\label{typeeps}
\Pbb[T_{X^nY^n} \in A_n]\geq\eps
\end{equation}
\item Lemma~\ref{typecovering} is satisfied with $\tau = \eps$ and $r = R_{\sf L} - R_0$.
\item The number of bits needed to express the joint type is negligible:
\begin{equation}
\label{jointtyperate}
\frac1n |\Xcal| |\Ycal| \log (n+1) \leq \eps.
\end{equation}
\end{enumerate}
To compress $x^n$ using side information ${y}^n$, the henchman first describes the joint type of $(x^n,{y}^n)$, then transmits the index of $x^n$ in the codebook $\Ccal(y^n, T_{x^n,{y}^n})$ that is guaranteed by Lemma~\ref{typecovering}. The description of the joint type only uses additional rate $\eps$ because the size of $\Tcal(\Xcal^n\times\Ycal^n)$ is bounded by $(n+1)^{|\Xcal||\Ycal|}$ and \eqref{jointtyperate} is satisfied. Therefore, for a given source sequence $x^n$ and side information sequence ${y}^n$, the henchman is able to send a message at rate $(R_{\sf L} - R_0)+\eps$ such that the eavesdropper can produce $z^n$ with distortion
\begin{equation}
\label{typedistortion}
d_{\sf E}(x^n,z^n) \leq D(R_{\sf L} - R_0+\eps, T_{x^n{y}^n}) + \eps.
\end{equation}
Now define
\begin{equation}
Q^*_{XY} \triangleq \argmax_{Q\in\Acal_n} D(R_{\sf L}-R_0+\eps,Q).
\end{equation}
From \eqref{typeeps}, we see that with probability at least $\eps$, the henchman and the eavesdropper can achieve distortion
\begin{IEEEeqnarray}{rCl}
d_{\sf E}(X^n,Z^n) &\leq& D(R_{\sf L} - R_0+\eps, T_{X^nY^n}) + \eps\\
&\leq& D(R_{\sf L} - R_0+\eps, Q^*_{XY}) + \eps
\end{IEEEeqnarray}
Therefore, in view of \eqref{contradict}, we can bound $D_{\sf E}$:
\begin{IEEEeqnarray}{rCl}
D_{\sf E} &\stackrel{(a)}{\leq}& D(R_{\sf L} - R_0+\eps, Q^*_{XY})+2\eps \\
&\stackrel{(b)}{\leq}& D(R_{\sf L} - R_0+\eps, P_{X}Q^*_{Y|X})+2\eps+o(\eps).
\end{IEEEeqnarray}
Step (a) follows from \eqref{contradict}. Step (b) is due to $\lVert Q^*_X - P_X\rVert_{\sf TV}<\eps$ and the fact that the rate-distortion function is continuous in $P_X$ with respect to total variation distance (e.g., see \cite{Palaiyanur2008}). Because $Q^*_{XY}\in\Acal_n$, we can also bound $R$ and $D_{\sf B}$. First, we have
\begin{IEEEeqnarray}{rCl}
R &\geq& I(Q^*_{XY}) - \eps\\
&\stackrel{(a)}{\geq}& I(P_XQ^*_{Y|X}) - \eps - o(\eps),
\end{IEEEeqnarray}
where (a) is due to the continuity of mutual information with respect to total variation distance. Next, we have
\begin{IEEEeqnarray}{rCl}
D_{\sf B} &\geq& \Ebb_{Q^*_{XY}}\,d_{\sf B}(X,Y) -\eps \\
&\stackrel{(a)}{\geq}& \Ebb_{P_XQ^*_{Y|X}}\,d_{\sf B}(X,Y) -\eps - o(\eps),
\end{IEEEeqnarray}
where (a) uses Property~\ref{tvproperties}c of total variation.

Assimilating the bounds that we have established, we can conclude that any achievable tuple $(R,R_0,R_{\sf L},D_{\sf B},D_{\sf E})$ lies in the region
\begin{equation}
\Scal_{\eps}\triangleq \bigcup_{P_{Y|X}}\left\{
 \begin{IEEEeqnarraybox}[][c]{rCl}
 (R,R_0,R_{\sf L},D_{\sf B},D_{\sf E}):\:R &\geq& I(X;Y)-o(\eps)\\
 D_{\sf B} &\geq& \Ebb\, d_{\sf B}(X,Y) - o(\eps)\\
 D_{\sf E} &\leq& \min\{D(R_{\sf L}), D(R_{\sf L} - R_0+\eps, P_{Y|X}) + o(\eps) \}
 \end{IEEEeqnarraybox}
\right\}.
\end{equation}
Since this holds for all $\eps >0$, we have
\begin{equation}
\label{intersectregions}
(R,R_0,R_{\sf L},D_{\sf B},D_{\sf E}) \in \bigcap_{\eps > 0} \Scal_{\eps}.
\end{equation}
The region in \eqref{intersectregions} is equal to the region in Theorem~\ref{mainresultlossy} (subject to $R_{\sf L} \geq R_0$), completing the converse proof.

\section{Achievability (lossy communication)}
\label{sec:lossyachievability}
In this section, we prove the achievability portion of Theorem~\ref{mainresultlossy}, the lossy communication counterpart to Theorem~\ref{mainresult}. The skeleton of the proof is similar to the one presented in Section~\ref{sec:losslessachievability}, but we will need some enhanced versions of some of the components. 

As in the lossless setting, we can view the henchman and the eavesdropper as the sender and receiver in a rate-limited system with side information $M$ (i.e., the public message) available to both parties. The correlation between the side information and the source sequence $X^n$ will govern the performance; therefore, we are interested in $P_{X^n|M=m}$ since this is the effective source distribution after accounting for common side information. As before, the encoder at Node A determines $P_{X^n|M=m}$ entirely. In Section~\ref{sec:losslessachievability}, we were motivated by the effect of random binning (which we later replaced with a likelihood encoder for ease of analysis). However, instead of simply randomly binning $X^n$ and using $K$ to hide the location within the bin, we now want to first perform lossy compression using a codebook of sequences from $\Ycal^n$, followed by a random binning of the codebook. Roughly speaking, this process results in a distribution $P_{X^n|M=m}$ that corresponds to selecting a $y^n$ sequence uniformly from a random codebook of size $2^{nR_0}$, then passing that sequence through a memoryless channel $\prod P_{X|Y}$. The justification for this assertion will become clear when we use a likelihood encoder later on; for now, we study the subproblem that just surfaced: lossy compression of a \emph{noisy version} of a codeword drawn uniformly from a random codebook. 

\subsection{Lossy compression of a noisy version of a codeword drawn uniformly from a random codebook}

Consider a codebook $c_{\ysf} = \{y^n(1),\ldots,y^n(2^{nR_{\sf C}})\}$ consisting of $2^{nR_{\sf C}}$ sequences in $\Ycal^n$. Select a codeword uniformly at random from $c_{\ysf}$ and denote it by $y^n(J)$, where $J \sim \text{Unif}[2^{nR_{\sf C}}]$. Pass $y^n(J)$ through a memoryless channel $\prod P_{X|Y}$ to produce a sequence $X^n$. An encoder describes $X^n$ using a noiseless link of rate $R$, and a decoder estimates it with a reconstruction sequence $Z^n$ (incurring distortion $d(X^n,Z^n))$. Both the encoder and decoder know the codebook $c_{\ysf}$, and together they constitute a $(n,c_{\ysf},R)$ code. The setup is shown in Figure~\ref{fig:codebooklossy} for a random codebook $\Ccal_{\ysf}$.

The following theorem generalizes Theorem~\ref{cbthm} (to recover that theorem, set $Y=X$).
\begin{thm}
\label{cblossythm}
Fix $P_{XY}$, $R$, $R_{\sf C}$ and $D$. Let $\Ccal_{\ysf}$ be a random codebook of $2^{nR_{\sf C}}$ codewords, each drawn independently according to $\prod_{i=1}^n P_Y(y_i)$. Let $\tau_n$ be any sequence that converges to zero sub-exponentially fast (i.e., $\tau_n=2^{-o(n)}$). If
\begin{equation}
\label{regime2}
R < \min\{R(D), R_Y(D) + R_{\sf C}\},
\end{equation}
then with high probability it is impossible to achieve distortion $D$ in the sense that
\begin{equation}
\label{cbobjlossy}
\lim_{n\to\infty}\Pbb_{\Ccal_{\ysf}}\Big[ \max_{(n,\Ccal_{\ysf},R)\text{ codes}}  \Pbb[d(X^n,Z^n)\leq D] > \tau_n \Big] = 0.
\end{equation}
The function $R_Y(D)$ is the rate-distortion function with side information:
\begin{equation}
R_Y(D) = \min_{P_{Z|XY}: \Ebb\, d(X,Z) \leq D} I(X;Z|Y).
\end{equation}
\end{thm}
Before diving into the proof, let us briefly justify why the regime in \eqref{regime2} is the one of interest. First, observe that whenever $R \geq R(D)$ is satisfied, distortion $D$ can be achieved by simply using a regular point-to-point rate distortion code to describe $X^n$. Second, whenever $R\geq R_Y(D) + R_{\sf C}$ holds, distortion $D$ can be achieved in roughly the following manner. The encoder first identifies a codeword in $\Ccal_{\ysf}$ that is jointly typical with $X^n$ (according to $P_{XY}$) and sends the index of the codeword using rate $R_{\sf C}$. The codeword is then treated as side information, which allows the encoder to describe $X^n$ using rate $R_Y(D)$. So we see that \eqref{regime2} is actually necessary for \eqref{cbobjlossy} to hold.

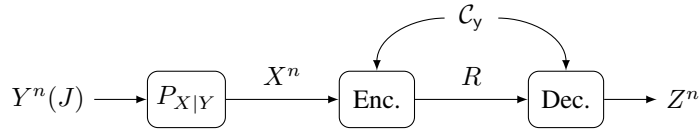
\begin{figure}
\begin{center}
\begin{tikzpicture}
 [node distance=1cm,minimum width=1cm,minimum height =.75 cm]
  \node[rectangle,minimum width=5mm] (codeword) {$Y^n(J)$};
  \node[node] (channel) [right =7mm of codeword] {$P_{X|Y}$};
  \node[node] (alice) [right =15mm of channel] {Enc.};
  \node[node] (bob) [right =1.5cm of alice] {Dec.};
  \node[coordinate] (dummy) at ($(alice.east)!0.5!(bob.west)$) {};
  \node[rectangle,minimum width=5mm] (xhat) [right =7mm of bob] {$Z^n$};
  \node[rectangle,minimum width=7mm] (key) [above =7mm of dummy] {$\Ccal_{\ysf}$};
  
  \draw [arw] (codeword) to (channel);
  \draw [arw] (channel) to node[minimum height=6mm,inner sep=0pt,midway,above]{$X^n$} (alice);
  \draw [arw] (alice) to node[minimum height=6mm,inner sep=0pt,midway,above]{$R$} (bob);
  \draw [arw] (bob) to (xhat);
  \draw [arw] (key) to [out=180,in=90] (alice);
  \draw [arw] (key) to [out=0,in=90] (bob);
 \end{tikzpicture}
 \caption{\small Lossy compression of a noisy version of a codeword drawn uniformly from a random codebook $\Ccal_{\ysf}=\{Y^n(1),\ldots,Y^n(2^{nR_{\sf C}})\}$. Both the encoder and decoder know the codebook $\Ccal_{\ysf}$. The encoder describes $X^n$, the output of a memoryless channel $\prod P_{X|Y}$ whose input is a randomly chosen codeword $Y^n(J)$, where $J\sim \text{Unif}[2^{nR_{\sf C}}]$.}
 \label{fig:codebooklossy}
 \end{center}
 \end{figure}

\begin{proof}
We follow the basic rubric of Section~\ref{sec:losslessachievability}, making modifications where they are needed. 

Fixing $P_{XY}$, we first restrict $(X^n,Y^n(J))$ to be jointly typical by writing
\begin{equation}
\label{typical2}
\Pbb[d(X^n,Z^n)\leq D] \leq \Pbb[d(X^n,Z^n)\leq D, \Acal] + \Pbb[\Acal^c],
\end{equation}
where $\Acal$ denotes the event $\{(X^n,Y^n(J)) \in \Tcal_{\delta}^n(X,Y) \}$. Note that the second term in~\eqref{typical2} vanishes in the limit for any $\delta>0$ since $(X^n,Y^n(J))$ is i.i.d. according to $P_{XY}$.

Continuing exactly as in Section~\ref{sec:losslessachievability}, we have
\begin{IEEEeqnarray}{rCl}
\max_{(n,\Ccal_{\ysf},R)\text{ codes}}  \Pbb[d(X^n,Z^n)\leq D,\Acal]
&\leq& 2^{nR} \max_{z^n\in \Zcal^n} \Pbb[ d(X^n,z^n)\leq D,\Acal]\\
&\stackrel{}{=}& 2^{nR} \max_{z^n\in \Zcal^n} \Ebb_J\,\Pbb[ d(X^n,z^n) \leq D, \Acal | Y^n(J)]\\
&=& 2^{-n(R_{\sf C} - R)} \max_{z^n\in \Zcal^n} \sum_{j=1}^{2^{nR_{\sf C}}} \Pbb[ d(X^n,z^n) \leq D, \Acal | Y^n(j)]
\end{IEEEeqnarray}
Denote the terms in the sum by $\zeta_{j,z^n}$:
\begin{IEEEeqnarray}{rCl}
\label{zetaj}
\zeta_{j,z^n} &\triangleq& \Pbb[ d(X^n,z^n) \leq D, \Acal | Y^n(j)]\\
&=& \sum_{x^n\in\Xcal^n} \prod_{i=1}^n P_{X|Y}(x_i | Y_i(j)) \cdot \mathbf{1}\{d(x^n,z^n)\leq D, (x^n,Y^n(j))\in \Tcal_{\delta}\}
\end{IEEEeqnarray}
Continuing in the manner of Section~\ref{sec:losslessachievability} leads us to
\begin{equation}
\label{finalprob2}
\Pbb\Big[ \max_{(n,\Ccal_{\ysf},R)\text{ codes}}  \Pbb[d(X^n,Z^n)\leq D,\Acal] > \tau_n \Big]
\leq |\Zcal|^n \max_{z^n\in\Zcal^n} \Pbb\Big[\sum_{j=1}^{2^{nR_{\sf C}}} \zeta_{j,z^n} > \tau_n 2^{n(R_{\sf C}-R)}\Big],
\end{equation}

As with the $\xi_{j,z^n}$ defined in Section~\ref{sec:losslessachievability} (Eq. \eqref{xij}), the $\zeta_{j,z^n}$ are i.i.d. due to the nature of the random codebook; however, they are no longer Bernoulli random variables. The following lemma, a straightforward generalization of Lemma~\ref{typebound}, shows that $\zeta_{j,z^n}$ is bounded above by $2^{-n(R_Y(D)-o(1))}$ with probability one. The proof is omitted.

\begin{lemma}
\label{typebound2}
Fix $P_{XY}$ and $y^n\in\Ycal^n$. If $X^n$ is distributed according to $\prod_{i=1}^n P_{X|Y=y_i}$, then for any $z^n$,
\begin{equation}
\Pbb\big[d(X^n,z^n)\leq D, (X^n,y^n) \in \Tcal_{\delta}^n \,\big|\, Y^n = y^n\big] \leq 2^{-n(R_Y(D)-o(1))},
\end{equation}
where $o(1)$ is a term that vanishes as $\delta\rightarrow 0$ and $n\rightarrow \infty$.
\end{lemma}

As mentioned, Lemma~\ref{typebound2} implies
\begin{equation}
\label{zetasupport}
\zeta_{j,z^n} \in [0,2^{-n(R_Y(D) - o(1))}].
\end{equation}
In addition to bounding the range of $\zeta_{j,z^n}$ we can also bound its expected value. In fact, the bound is the same as for $\xi_{j,z^n}$.
\begin{IEEEeqnarray}{rCl}
\Ebb_{\Ccal_{\ysf}} \zeta_{j,z^n}
&=& \Ebb_{\Ccal_{\ysf}}\, \Pbb[ d(X^n,z^n) \leq D, \Acal \,|\, Y^n(j)] \\
&\leq& \Ebb_{\Ccal_{\ysf}}\, \Pbb[ d(X^n,z^n) \leq D, X^n \in \Tcal_{\delta} \,|\, Y^n(j)]\\
&=& \Pbb[d(X^n,z^n) \leq D, X^n \in \Tcal_{\delta} ] \qquad (\text{where }X^n \sim \prod P_X)\\
\label{zetamean}&\stackrel{(a)}{\leq}& 2^{-n (R(D) - o(1))},
\end{IEEEeqnarray}
where (a) is due to Lemma~\ref{typebound}.

We are now ready to apply a Chernoff bound to the probability in \eqref{finalprob2}. First, we extend the Chernoff bound in Lemma~\ref{chernoff} to random variables taking values on the interval $[0,a]$ instead of just binary random variables.
\begin{cor}
\label{chernoffcor2}
If $X^m$ is a sequence of i.i.d. random variables on the interval $[0,a]$ with $\Ebb[X_i] = p$, then
\begin{equation}
\Pbb\Big[\sum_{i=1}^m X_i > k\Big] \leq \left(  \frac{e\!\cdot\! m \!\cdot\! p}{k} \right )^{k/a}.
\end{equation}
\end{cor}

\begin{proof}
We start by proving the case $a=1$. To begin, we claim that if $X\in[0,1]$ and $Y\in\{0,1\}$ are random variables such that $\Ebb[X] = \Ebb[Y]$ and $f: [0,1] \rightarrow \Rbb$ is convex, then
\begin{equation}
\Ebb[f(X)] \leq \Ebb[f(Y)].
\end{equation}
To see this, observe that for $x\in[0,1]$,
\begin{equation}
f(x) \leq x f(1) + (1-x) f(0).
\end{equation}
Taking expectations gives
\begin{IEEEeqnarray}{rCl}
\Ebb[f(X)]  &\leq& \Ebb[X] f(1) + (1-\Ebb[X]) f(0) \\
&=& \Ebb[Y] f(1) + (1-\Ebb[Y]) f(0) \\
&=& \Ebb[f(Y)],
\end{IEEEeqnarray}
verifying the claim. Now, since $f(x) = e^{\lambda x}$ is convex, the inequality $\Ebb[e^{\lambda Y}] \leq \Ebb[e^{\lambda X}]$ holds and can be applied to the proof of Lemma~\ref{chernoff} at \eqref{chernoffcorstep}.

With the case $a=1$ shown, we now consider any $a>0$. If we let $Y_i = \frac{1}{a} X_i \in [0,1]$, then the previous case applies and we have
\begin{IEEEeqnarray}{rCl}
\Pbb\Big[\sum_{i=1}^m X_i > k\Big] &\leq& \Pbb\Big[\sum_{i=1}^m a Y_i > \frac{k}{a}\Big] \\
&\leq& \Big ( \frac{e\!\cdot\! m \!\cdot\! \Ebb[Y_1]}{k/a} \Big )^{k/a} \\
&=& \Big ( \frac{e\!\cdot\! m \!\cdot\! p}{k} \Big )^{k/a}.
\end{IEEEeqnarray}
\end{proof}

Using the support bound and the expected value bound in \eqref{zetasupport} and \eqref{zetamean}, we can apply Corollary~\ref{chernoffcor2} to the probability in \eqref{finalprob2} by identifying
\begin{IEEEeqnarray}{rCl}
m &=& 2^{n R_{\sf C}}\\
a &=& 2^{-n(R_Y(D) - o(1))}\\
p &\leq& 2^{-n(R(D)-o(1))}\\ 
k &=& \tau_n 2^{n (R_{\sf C}-R)}.
\end{IEEEeqnarray}
This gives
\begin{equation}
\label{doubleexp2}
\Pbb\Big[\sum_{j=1}^{2^{nR_{\sf C}}} \zeta_{j,z^n} > \tau_n 2^{n(R_{\sf C}-R)}\Big] \leq 2^{-n\alpha2^{n\beta}},
\end{equation}
where
\begin{IEEEeqnarray}{rCl}
\alpha &=& R(D) - R -o(1)\\
\beta &=& R_{\sf C}+R_Y(D)-R-o(1).
\end{IEEEeqnarray}
For small enough $\delta$ and large enough $n$, both $\alpha$ and $\beta$ are positive and bounded away from zero, and \eqref{doubleexp2} vanishes doubly exponentially fast. Consequently, the expression in \eqref{finalprob2} vanishes, completing the proof of Theorem~\ref{cblossythm}.
\end{proof}

The following corollary to Theorem~\ref{cblossythm} is immediate, and, as in Section~\ref{sec:losslessachievability}, will serve as the bridge between the subproblem we have been considering and the henchman problem. 

 \begin{cor}
\label{cbcorlossy}
Fix $P_{XY}$. If $R < \min\{R(D), R_Y(D) + R_{\sf C}\}$, then
\begin{equation}
\lim_{n\to\infty}\Ebb_{\Ccal_{\ysf}}\Big[ \min_{(n,\Ccal_{\ysf},R)\text{ codes}}  \Pbb[d(X^n,Z^n)\geq D]\Big] = 1.
\end{equation}
\end{cor}

\subsection{Likelihood encoder}
Returning to the henchman problem, we follow the basic structure of Section~\ref{sec:losslessachievability}. Fixing $P_{Y|X}$ (and thus a joint distribution $P_{XY}$), consider a codebook $c=\{y^n(m,k)\}$ of $2^{n(R+R_0)}$ sequences from $\Ycal^n$ and define a likelihood encoder for this codebook by
\begin{equation}
P_{M|X^nK}(m|x^n,k) \propto \prod_{i=1}^n P_{X|Y}(x_i | y_i(m,k)),
\end{equation}
where $\propto$ indicates that appropriate normalization is required. The distribution $P_{X^nMK}$ induced by using this encoder with a random codebook is intimately related to an idealized distribution $Q_{X^nMK}$ defined by
\begin{equation}
Q_{X^nMK}(x^n,m,k) \triangleq 2^{-n(R+R_0)} \prod_{i=1}^n P_{X|Y}(x_i | y_i(m,k)).
\end{equation}
Indeed, just as in Lemma~\ref{distrclose}, one can use the soft covering lemma to show that if $R > I(X;Y)$, then
\begin{equation}
\label{distrtvlossy}
\lim_{n\to\infty} \Ebb_{\Ccal} \big\lVert P_{X^nMK} - Q_{X^nMK} \big\rVert_{\sf TV} = 0,
\end{equation}
where $\Ccal$ is a random codebook with each codeword drawn independently according to $\prod_{i=1}^n P_Y$.

Inspecting $Q_{X^nMK}$ reveals that $Q_{X^n|M=m}$ is exactly the distribution that was addressed in the recent interlude. To see this, observe that $Q_{X^nK|M=m}$ is the joint distribution that arises from selecting a codeword uniformly from a codebook of size $2^{nR_0}$ and passing it through a memoryless channel $\prod P_{X|Y}$. To be explicit,
\begin{equation}
\label{idealconditional}
Q_{X^n|M}(x^n | m) = 2^{-nR_0} \sum_{k=1}^{2^{nR_0}} \prod_{i=1}^n P_{X|Y}(x_i | Y_i(m,k)).
\end{equation}
Proceeding with the analysis of the eavesdropper's distortion, first denote the event
\begin{equation}
\Ecal = \{ d(X^n, z^n(M,M_{\sf H})) \geq \pi(R_{\sf L},R_0,P_{Y|X}) - \eps \},
\end{equation}
where
\begin{equation}
\pi(R_{\sf L},R_0,P_{Y|X}) \triangleq 
\begin{cases}
D(R_{\sf L}) & \text{if } R_0 > R_{\sf L}\\
\min\{D(R_{\sf L}), D_Y(R_{\sf L} - R_0)\} & \text{if }R_0 \leq R_{\sf L}
\end{cases}
\end{equation}
The purpose of $\pi(\cdot)$ is to treat the cases $R_{\sf L} < R_0$ and $R_{\sf L} \geq R_0$ concurrently.

Taking the expectation of \eqref{listsecrecylossy} with respect to a random codebook, we have
\begin{equation}
\label{relatecorlossy} \Ebb_{\Ccal} \Big[\min_{\substack{m_{\sf H}(x^n,m),z^n(m,m_{\sf H}):\\|\Mcal_{\sf H}|\leq2^{nR_{\sf L}}}} \Pbb_{P_{X^nM}}
[\Ecal]\Big] = \Ebb_M \,\Ebb_{\Ccal} \Big[\min_{\substack{m_{\sf H}(x^n),z^n(m_{\sf H}):\\|\Mcal_{\sf H}|\leq2^{nR_{\sf L}}}} \Pbb_{Q_{X^n|M}} [\Ecal | M] \Big] + o(1).
\end{equation}
From \eqref{idealconditional}, we see that the expression in \eqref{relatecorlossy} is exactly what is addressed by Corollary~\ref{cbcorlossy} after identifying $(R_0,R_{\sf L})$ with $(R_{\sf C}, R)$. Note that we are invoking the corollary with $D=\pi(R_{\sf L},R_0,P_{Y|X}) - \eps$, which means that 
\begin{equation}
R_{\sf L} < \min\{ R(D), R_Y(D) + R_0 \}.
\end{equation} 
Thus, we have
\begin{equation}
\lim_{n\to\infty}\Ebb_{\Ccal} \Big[\min_{\substack{m_{\sf H}(x^n),z^n(m_{\sf H}):\\|\Mcal_{\sf H}|\leq2^{nR_{\sf L}}}} \Pbb_{Q_{X^n|M=m}} [\Ecal | M=m] \Big] = 1.
\end{equation}
Therefore, we can conclude that there exists a codebook such that the associated likelihood encoder ensures~\eqref{listsecrecylossy}, because \eqref{listsecrecylossy} holds when averaged over random codebooks. 

We now complete the proof of achievability by showing that the likelihood encoder can be used to achieve distortion $\Ebb\, d_{\sf B}(X,Y)$ at the legitimate receiver (this is also done in \cite{Song2014}). To do this, Node B uses a deterministic decoder that simply produces the codeword indexed by $(m,k)$, i.e.,
\begin{equation}
P_{Y^n|MK}(y^n|m,k) = \mathbf{1}\{ y^n = y^n(m,k) \}.
\end{equation}
Defining $Q_{X^nMKY^n} \triangleq Q_{X^nMK}P_{Y^n|MK}$, we can write
\begin{IEEEeqnarray}{rCl}
\Ebb_{\Ccal} \big\lVert P_{X^nY^n} - Q_{X^nY^n} \big\rVert_{\sf TV}
&\stackrel{(a)}{\leq}& \Ebb_{\Ccal} \big\lVert P_{X^nMK}P_{Y^n|MK} - Q_{X^nMK} P_{Y^n|MK} \big\rVert_{\sf TV} \\
&\stackrel{(b)}{=}&  \Ebb_{\Ccal} \big\lVert P_{X^nMK} - Q_{X^nMK} \big\rVert_{\sf TV}\\
\label{lossyjointclose} &\stackrel{(c)}{\rightarrow}& 0,
\end{IEEEeqnarray}
where (a) and (b) follow from Properties~\ref{tvproperties}d and \ref{tvproperties}c, and (c) is due to \eqref{distrtvlossy}. Now notice that $\Ebb_{\Ccal} Q_{X^nY^n}$ is exactly the product distribution $\prod_{i=1}^n P_{XY}$ (a fact which is straightforward to verify). Therefore, by \eqref{lossyjointclose} and the weak law of large numbers, we have
\begin{IEEEeqnarray}{rCl}
\Ebb_{\Ccal}\,\Pbb\Big[d_{\sf B}(X^n,Y^n) > \Ebb\,d_{\sf B}(X,Y)+\eps\Big]
&=& \Ebb_{\Ccal}\,\Pbb_{Q_{X^nY^n}} \Big[d_{\sf B}(X^n,Y^n) > \Ebb\,d_{\sf B}(X,Y)+\eps\Big] + o(1)\\
&=& o(1)
\end{IEEEeqnarray}
This completes the achievability portion of the proof of Theorem~\ref{mainresultlossy}.


\appendices
\section{Proof of Lemma~\ref{typebound}}
We first bound $\Pbb[d(X^n,z^n)\leq D]$ and resolve the event $X^n\in \Tcal_{\delta}^n$ afterward. We use the V-shell notation from the method of types \cite{Csiszar1998}: for a stochastic matrix $V_{Z|X}$, the set of $z^n$ sequences having conditional type $V$ is denoted by $\Tcal_{V}^n(x^n)$. Note that all pairs $(x^n,z^n)$ satisfying $z^n\in \Tcal_{V}^n(x^n)$ have the same joint type (denoted by $T_{x^nz^n}$).

Diving in, we have
\begin{IEEEeqnarray}{rCl}
\Pbb[d(X^n,z^n)\leq D] &=& \sum_{x^n\in\Xcal^n} P_{X^n}(x^n) \mathbf{1}\{d(x^n,z^n)\leq D\} \\
 &\stackrel{(a)}{=}& \sum_{V_{X|Z}} \sum_{x^n\in \Tcal_{V}^n(z^n)} P_{X^n}(x^n)  \mathbf{1}\{d(x^n,z^n)\leq D\}\\
 &\stackrel{(b)}{=}& \sum_{V_{X|Z}} \sum_{x^n\in \Tcal_{V}^n(z^n)} 2^{-n(D(T_{x^n} || P_X)+H(T_{x^n}))} \mathbf{1}\{\Ebb_{T_{x^nz^n}} d(X,Z)\leq D\}\\
 &\stackrel{(c)}{\leq} & \sum_{\substack{V_{X|Z}:\\ \Tcal_{V}^n(z^n)\neq\emptyset}} 2^{nH(T_{x^n} | T_{z^n})} 2^{-n(D(T_{x^n} || P_X)+H(T_{x^n}))} \mathbf{1}\{\Ebb_{T_{x^nz^n}} d(X,Z)\leq D\}\\
 &=& \sum_{\substack{V_{X|Z}:\\ \Tcal_{V}^n(z^n)\neq\emptyset}} 2^{-n(I(T_{x^nz^n}) + D(T_{x^n} || P_X))} \mathbf{1}\{\Ebb_{T_{x^nz^n}} d(X,Z)\leq D\}\\
 \label{exponent1} &\stackrel{(d)}{\leq}& \exp \Big\{ - n \min_{V: \Ebb_{T_{x^nz^n}} d(X,Z)\leq D} [I(T_{x^nz^n}) + D(T_{x^n} || P_X)] + O(\log n) \Big\}.
\end{IEEEeqnarray}
In step (a), we partition the set $\Xcal^n$ according to the conditional type of $x^n$ given $z^n$. Step (b) follows by observing that the summands only depend on the joint type of $(x^n,z^n)$. Step (c) uses a bound on the size of $\Tcal_{V}(z^n)$, and step (d) follows from the fact that the number of conditional types is polynomial in $n$.

We can continue by lower bounding the first term in the (normalized) exponent of \eqref{exponent1}:
\begin{IEEEeqnarray}{rCl}
\IEEEeqnarraymulticol{3}{l}{\nonumber
\min_{V: \Ebb_{T_{x^nz^n}} d(X,Z)\leq D} I(T_{x^nz^n}) + D(T_{x^n} || P_X)  
}\\
\qquad\qquad &\geq& \min_{z^n}\min_{V: \Ebb_{T_{x^nz^n}} d(X,Z)\leq D} I(T_{x^nz^n}) + D(T_{x^n} || P_X) \\
&=& \min_{Q_{XZ}:\Ebb_{Q} d(X,Z) \leq D} I_Q(X;Z) + D(Q_X||P_X) \\
&=& \min_{Q_X} \min_{Q_{Z|X}: \Ebb_{Q} d(X,Z) \leq D} I_Q(X;Z) + D(Q_X||P_X) \\
&=& \min_{Q_X}\, [R(D,Q_X) + D(Q_X||P_X)],
\end{IEEEeqnarray}
where $R(D,Q_X)$ denotes the rate-distortion function for a source $Q_X$.

So far, we have shown that the following holds for all $z^n$:
\begin{equation}
\Pbb[d(X^n,z^n)\leq D] \leq \exp \{ -n \cdot \min_{Q_X}\, [R(D,Q_X) + D(Q_X||P_X)] + O(\log n) \}.
\end{equation}

However, this is not quite the bound we seek; a simple example will reveal that it is possible to have
\begin{equation}
\min_{Q_X}\, [R(D,Q_X) + D(Q_X||P_X)] < R(D).
\end{equation}
Indeed, consider $P_X \sim \text{Bern}(p),p\in(D,1/2)$ and $Q_X \sim \text{Bern}(q)$. After simplifying, we find that
\begin{equation}
R(D,Q_X) + D(Q_X||P_X) = q\log \frac{1}{p} + (1-q) \log \frac{1}{1-p} - h(D).
\end{equation}
Minimizing this expression over $q\in[0,1]$ gives
\begin{IEEEeqnarray}{rCl}
\min_{Q_X}\, [R(D,Q_X) + D(Q_X||P_X)] &=& \min\Big\{\log\frac{1}{p},\log\frac{1}{1-p}\Big\} - h(D)\\
 &<& h(p) - h(D)\\
 &=& R(D).
\end{IEEEeqnarray}
To resolve this issue, we introduce the event $X^n\in \Tcal_{\delta}^n$ into the expression we want to bound. Modifying the steps above accordingly, we have
\begin{IEEEeqnarray}{rCl}
-\frac{1}{n}\log \Pbb[d(X^n,z^n)\leq D, X^n \in \Tcal_{\delta}^n] &\geq& \min_{Q_X: \lVert Q_X - P_X \rVert_{\sf TV} < \delta} R(D,Q_X) + D(Q_X||P_X) - O(\tfrac{\log n}{n})\\
&\stackrel{(a)}{\geq}& \min_{Q_X: \lVert Q_X - P_X \rVert_{\sf TV} < \delta} R(D,Q_X) - O(\tfrac{\log n}{n})\\
&\stackrel{(b)}{=}& R(D) - O(\delta \log \tfrac{1}{\delta}) - O(\tfrac{\log n}{n})\\
&=& R(D) - o(1),
\end{IEEEeqnarray}
where step (a) is due to the non-negativity of relative entropy and step (b) follows from the uniform continuity of the rate-distortion function with respect to total variation distance (e.g., \cite{Palaiyanur2008}).
\bibliographystyle{IEEEtran}
\bibliography{henchman_isit}
\end{document}